\documentclass[%
 reprint,
 amsmath,amssymb,
 aps,
pra,
]{revtex4-1}
\usepackage{float}
\usepackage[normalem]{ulem}
\usepackage{graphicx}
\usepackage{dcolumn}
\usepackage{bm}
\usepackage[dvipsnames]{xcolor}
\usepackage{amsmath,amssymb,amsthm}
\usepackage{bbold} 

\usepackage{multirow}
\usepackage{enumitem}
\usepackage{algorithm}
\usepackage{algpseudocode}
\usepackage{array}
 
\makeatletter 
\renewcommand{\fnum@figure}{\textbf{Fig.~\thefigure}}
\renewcommand{\fnum@table}{\textbf{Table~\thetable}}
\makeatother 
 
\newtheorem{theorem}{Theorem}[section]
\newtheorem{lemma}[theorem]{Lemma}

\newtheorem{innercustomgeneric}{\customgenericname}
\providecommand{\customgenericname}{}
\newcommand{\newcustomtheorem}[2]{%
  \newenvironment{#1}[1]
  {%
   \renewcommand\customgenericname{#2}%
   \renewcommand\theinnercustomgeneric{##1}%
   \innercustomgeneric
  }
  {\endinnercustomgeneric}
}
\newcustomtheorem{LabelTheorem}{Theorem}

\newcommand{\ket}[1]{\ensuremath{\vert{#1}\rangle}}
\newcommand{\bra}[1]{\ensuremath{\langle{#1}\vert}}

\newcommand{\Tr}{\ensuremath{\mathrm{Tr}}}

\newcommand{\sfrac}[2]{\ensuremath{{\textstyle\frac{#1}{#2}}}}

\newcommand{\half}[0]{\sfrac{1}{2}}

\newcommand{\norm}[1]{\ensuremath{\Vert{#1}\Vert}}
\newcommand{\dnorm}[1]{\ensuremath{\Vert{#1}\Vert_{\diamond}}}

\renewcommand{\vec}[1]{{\mathbf #1}}

\newcommand{\TJWat}{IBM Quantum, IBM T.J.~Watson Research Center, Yorktown Heights, NY 10598, USA}

\newcommand{\sref}[2]{Supplementary Materials Sec.~{#1}}

\begin{document}

\title{Probabilistic error cancellation with sparse Pauli-Lindblad models on noisy quantum processors} 

\author{Ewout van den Berg, Zlatko K. Minev, Abhinav Kandala, Kristan Temme}

\affiliation{\TJWat}
\date{\today}


\maketitle 

{\bf{
Noise in pre-fault-tolerant quantum computers can result in biased
estimates of physical observables. Accurate bias-free estimates can be
obtained using probabilistic error cancellation (PEC), which is
an error-mitigation technique
that effectively inverts well-characterized noise channels. Learning
correlated noise channels in large quantum circuits, however, has been
a major challenge and has severely hampered experimental realizations.
Our work presents a practical protocol for learning and inverting a
sparse noise model that is able to capture correlated noise and scales
to large quantum devices. These advances allow us to demonstrate PEC
on a superconducting quantum processor with crosstalk errors, thereby
providing an important milestone in opening the way to quantum
computing with noise-free observables at larger circuit volumes.}}

\paragraph*{Introduction}
As a result of continuous improvement in quantum hardware and control
systems, quantum processors are now able to provide more qubits with
longer coherence times and better gate
fidelities~\cite{zhang2020ibm,arute2019,Wu2021}. Despite these
improvements, the levels of noise in current quantum processors still
limit the depth of quantum circuits and reduce the accuracy of
measured observables. Nevertheless, there is a growing number of
quantum applications that run on noisy quantum processors and still
provide competitive
results~\cite{peruzzo2014variational,kandala2017hardware,KIM2021WYMa-arXiv,havlivcek2019supervised,schuld2019quantum}. Fault
tolerance using quantum error correction or similar techniques would
solve many noise related issues, but until this is achieved, quantum
error
mitigation~\cite{TEM2017BGa,LiBenjamin2017,kandala2019error,END2018BLa}
may very well be the best way forward. Unlike error correction, which
ensures that quantum circuits can be executed faithfully, error
mitigation only aims to produce accurate expectation values
$\langle A \rangle$ of observables $A$.

One of the earliest and most general protocols for error mitigation is
probabilistic error cancellation (PEC)~\cite{TEM2017BGa}. To implement
the error mitigated action $\mathcal{U}(\rho) = U \rho \, U^{\dag}$ of
an ideal gate $U$ on a devices where only noisy operations
${\cal U} \circ \Lambda$ are available, the protocol first requires an
accurate noise model $\Lambda$. The action of the ideal gate would
then be obtained by applying the mathematical inverse $\Lambda^{-1}$
before the noisy gate. Although $\Lambda^{-1}$ is not a physical
operation, it can be expressed as a linear combination of gates and
state-preparation operations~\cite{TEM2017BGa,END2018BLa}. The PEC
protocol implements this linear combination on average by promoting it
to a quasi-probability distribution. Sampling the distribution
generates physical circuit instances and results in an expectation
value $\langle \hat{A}_N \rangle$ that is unbiased and completely
removes the effect of $\Lambda$. However, this comes at the expense of
an increased sampling overhead we denote by $\gamma$, which captures
the noise strength and the resulting increase of the standard
deviation.

Despite the method's theoretical
appeal~\cite{END2018BLa,GUO2022Ya-arXiv,PIV2021SWa-arXiv,doi:10.7566/JPSJ.90.032001,piveteau2021error,PhysRevResearch.3.033178,takagi2021fundamental},
practical challenges have limited its demonstration to the one- and
two-qubit level~\cite{doi:10.1126/sciadv.aaw5686,Zha2020LZCa}.  The
main difficulty has been the accurate representation of the noise in a
full device, which is particularly complicated by cross-talk errors
that occur during the parallel application of gates.  This has lead to
protocols where a quasi-probability distribution for mitigation is
determined by minimizing the deviation of a set of measured and exact
expectation values~\cite{strikis2021learning}.  Fully scalable
implementations of PEC require a noise model $\Lambda$ that accurately
captures correlated errors across all qubits, has a compact
representation that can be learned efficiently, and has an inverse
representation that enables tractable sampling from the associated
quasi-probability distribution.

We address these challenges in the context of quantum circuits that
consist of $l$ layers of noisy two-qubit gates interleaved with layers
of single-qubit gates. Each layer $i=1,\ldots,l$ consists of a noisy
operator $\tilde{\mathcal{U}}_i$ and is error mitigated by
$\Lambda_i^{-1}$, as shown in Fig.~\ref{Fig:Layers}a. The noise
channel $\Lambda_i$ is specific to the gates in layer $i$ and is
assumed to be a Pauli channel. If needed, this can be ensured using
Pauli
twirling~\cite{PhysRevLett.76.722,knill2004fault,kern2005quantum,geller2013efficient,wallman2016noise},
as illustrated in Figs.~\ref{Fig:Layers}b and~\ref{Fig:Layers}c for an
example with four qubits and two {\sc{cx}} gates.

We present an efficient mitigation scheme that models the noise across
each layer of two-qubit gates as a sparse Pauli-Lindblad error model.
In our experiments, the model includes only weight-one and weight-two
Pauli terms whose support coincides with the quantum processor's
connectivity. The parameters of the resulting model scale linearly
with the number of qubits, which ensures that the model is efficiently
represented and easy to learn. The inverse noise model is obtained
simply by negating the model coefficients and gives rise to a
quasi-probability distribution on Pauli matrices. We provide an
efficient algorithm for sampling this distribution in linear time with
the number of model coefficients. The mitigation Paulis can be
combined with those used for twirling as well as with the single-qubit
operations in the interleaved layers. The error mitigation scheme
therefore maintains the original circuit structure and changes only
the classical distribution of the single-qubit gates.

\begin{figure*}[t]
\setlength{\tabcolsep}{10pt}
\centering
\includegraphics[width=0.98\textwidth]{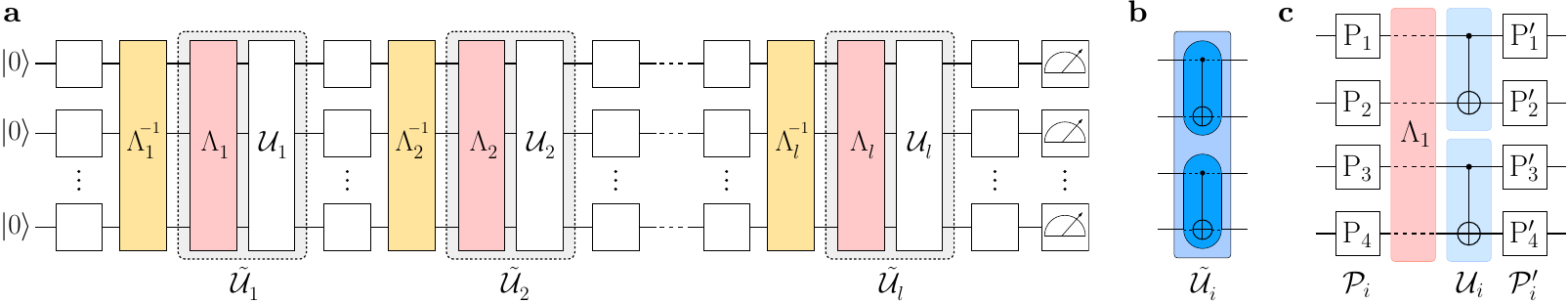}
\caption{ \textbf{Context of the noise model.} ({\bf{a}}) ideal error
  mitigation of a circuit consisting of $l$ layers of noisy two-qubit
  gates interleaved with layers of single-qubit gates. ({\bf{b}})
  example of a layer consisting of two noisy {\sc{cx}}
  gates. ({\bf{c}}) expansion of the same layer in terms of the ideal
  gates $\mathcal{U}_i$ and noise channel $\Lambda_i$, flanked with
  Pauli-twirl gates $\mathcal{P}_i$ and
  $\mathcal{P}_i' = \mathcal{U}_i \mathcal{P}_i\mathcal{U}_i^{\dag}$, where $\mathcal{P}_i$ is sampled uniformly at random.}\label{Fig:Layers}
\end{figure*}

\paragraph*{Pauli-Lindblad noise model}

We model a given $n$-qubit Pauli noise channel $\Lambda$ that arises
from a sparse set of local interactions, according to a Lindblad
Master equation~\cite{BRE2002Pa} with generator
$\mathcal{L}(\rho) = \sum_{k \in \mathcal{K}} \lambda_k\big(P_k\rho
P_k^{\dag} - \rho\big)$, where $\mathcal{K}$ represents a set of local
Paulis $P_k$ and $\lambda_k$ denotes the corresponding model
coefficient. The resulting model is then given by
(see~\sref{SIII}{Sec:ErrorModel})
\begin{equation}\label{Eq:MainFinalForm}
\Lambda(\rho) = \mbox{exp}[\mathcal L](\rho) = \prod_{k \in \mathcal{K}}\left(w_k\cdot + (1-w_k)P_k\cdot P_k^{\dag}\right)\rho,
\end{equation}
where~$w_k = 2^{-1}(1 + e^{-2 \lambda_k})$. The model
terms~$\mathcal{K}$ are chosen to reflect the noise interactions in
the quantum processor and their number, which determines the model
complexity and expressivity, typically scales polynomially in $n$ and
therefore allows us to represent noise models for the full device by a
small set of nonnegative coefficients~$\lambda_k$.

The fidelity of a Pauli matrix $P_b$ with respect to $\Lambda$ is
given by $f_b = \frac{1}{2^n}\Tr\big(P_{b}^{\dag}
\Lambda(P_b)\big)$. Defining the symplectic inner product
$\langle b,k\rangle_{sp}$ to be $0$ if Paulis $P_b$ and $P_k$ commute
and $1$ otherwise, we can concisely express the relationship between
model coefficients $\lambda$ and the vector
$f=\{f_b\}_{b\in\mathcal{B}}$ of fidelities for an arbitrary set of
Paulis $\mathcal{B}$ as
$\log(f) = -2M(\mathcal{B},\mathcal{K})\lambda$, where the logarithm
is applied elementwise and the entries of binary matrix
$M(\mathcal{B},\mathcal{K})$ are given by
$M_{b,k} = \langle b,k\rangle_{sp}$. For a given $\lambda$ this allows
us to evaluate the fidelity of any set of Paulis $\mathcal{B}$. More
importantly, though, the relationship allows us to fit physical model
parameters, $\lambda\geq 0$, given the fidelity estimates $\hat{f}$
for a set of benchmark Paulis $\mathcal{B}$ by solving a nonnegative
least-squares problem in $\log(\hat{f})$;
see~\sref{SIII.3}{subsec:learning} for more details.

Various methods of learning the fidelities of Pauli channels are
known~\cite{FLA2020Wa,ERH2019WPMa,PhysRevX.4.011050,HEL2019XVWa} and
have been implemented experimentally~\cite{harper2020efficient}. The
central idea in these methods is that the same noise process is
repeated up to $d$ times and the corresponding Pauli expectation
values are measured at every depth. The fidelities for the noise
channel can then be extracted from the decay rates in the resulting
curves in a way that is robust to state-preparation and measurement
(SPAM) errors. In \sref{SIV.2}{Sec:ErrorAnalysis} we provide
theoretical guarantees for the sample complexity for learning the
error model. Under mild conditions on the minimal fidelity of the
noise channel and the level of SPAM errors we provide the following
result for all the fidelities predicted by the model: Assume that the
channel can be represented with the model Paulis from set
$\mathcal{K}$, and that the channel fidelities for Paulis in
$\mathcal{B}$ are learned by benchmarking up to depth $d$ with at
least $2\epsilon^{-2}\log(2\vert\mathcal{B}\vert/\delta)$ circuit
instances for each of the relevant measurement bases. Then it holds
with probability at least $1-\delta$ that the estimates $\hat{f}_j$ of
all fidelities $f_j$ are bounded by
\begin{equation}\label{Eq:MainBound1}
C_{\epsilon}^{-\tau} \leq f_j \, \hat{f_j}^{-1} \leq C_{\epsilon}^{\tau},
\end{equation}
with
$\tau\! =\!
\sqrt{\vert\mathcal{K}\vert\cdot\vert\mathcal{B}\vert}/(\sigma_{\min}(M(\mathcal{B},\mathcal{K}))d)$,
and $C_{\epsilon}=\big({\textstyle\frac{1+4\epsilon}{1-4\epsilon}}\big)$.

\begin{figure*}[!t]
\centering
\includegraphics[width=0.975\textwidth]{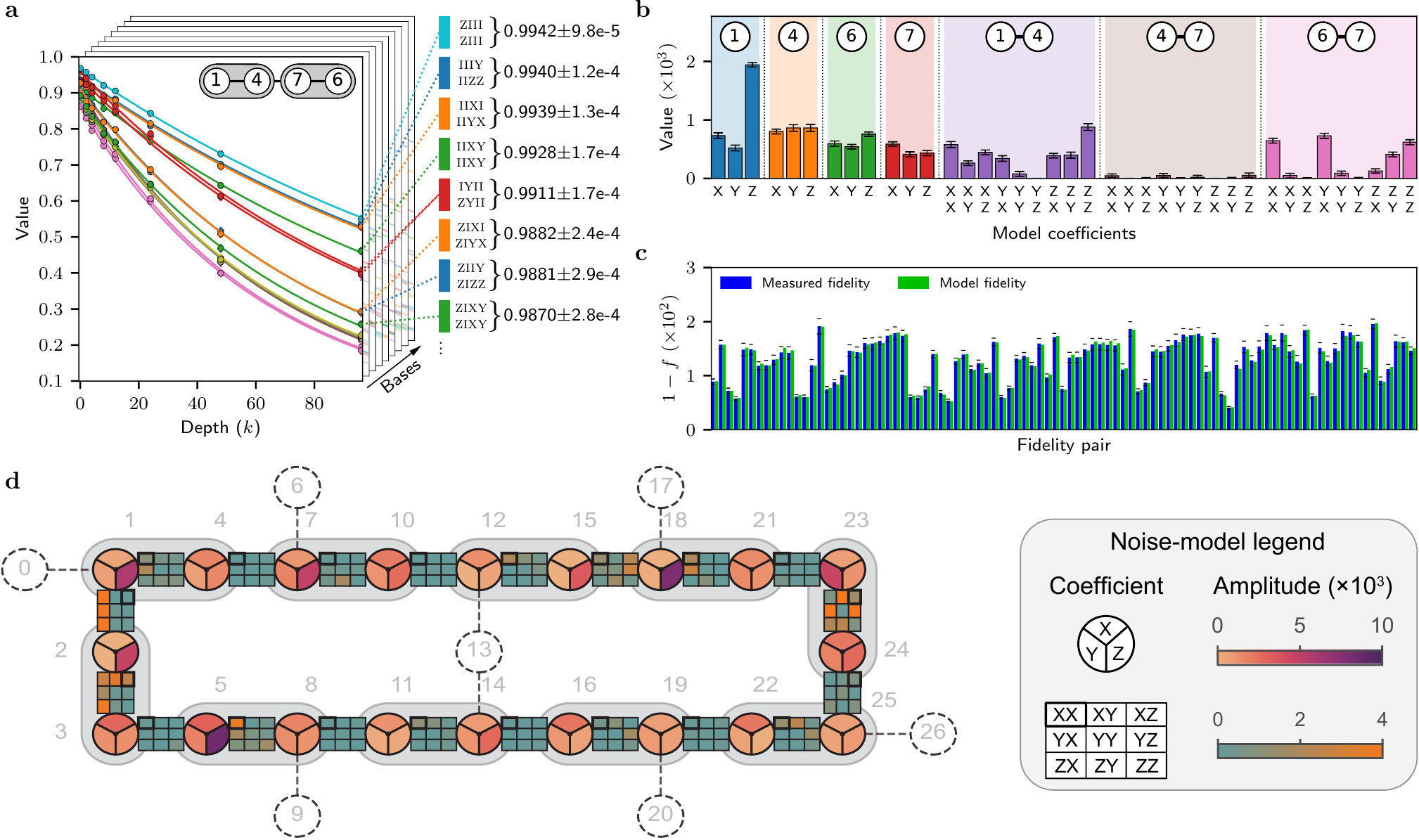}
\caption{ \textbf{Learning the noise channel.}  (\textbf{a}) The first
  step in learning our noise model, in this case for the four-qubit
  layer depicted in the top inset with two concurrent {\sc{cx}} gates,
  is to measure a set of observables with increasing numbers of
  circuit repetitions $k$ (even) up to some maximum depth $d$. This
  requires measurements in nine different bases, illustrated by the
  stacked planes. Associated with each observable $P_o$ is a fidelity
  of the form $\alpha_{o}(f_1f_2)^{k/2}$, where $\alpha_{o}$ is a
  constant that captures the state-preparation and measurement error,
  and $f_1$ and $f_2$ are the fidelities of the noise channel for two
  Pauli terms. We estimate the values of the different fidelity pairs
  in a consistent manner by fitting exponentially decaying curves
  through the data point of all observables that include the same
  pair, which may arise in multiple bases and possibly different
  observables, whose curves may have a different offset values
  $\alpha_o$. The legend on the right hand side illustrates the
  fidelity estimates for several pairs along with their standard
  deviation obtained using a 100-fold bootstrap (the error bars for
  the data points are small and largely covered by the markers).
  (\textbf{b}) Model coefficients obtained using a nonnegative
  least-squares fit of the log fidelities.  (\textbf{c}) Plot of one
  minus the fidelity for each of the measured fidelity pairs including
  error bars representing the standard deviation (vertical lines in
  the error bars are omitted for clarity), along with the
  corresponding fidelities from the learned noise model.  (\textbf{d})
  Visualization of the sparse noise model of a 20-qubit layer with 10
  concurrent {\sc{cx}} gates (shaded pairs) overlaid on the topology
  of the \textsc{ibm\_hanoi} quantum processor. Circles denote qubits
  (labeled by numbers); colored wedges in the circle visualize the
  single-body~$X,Y$, and~$Z$ Lindblad coefficients (see legend
  top). Two-body coefficients, e.g. $XX$, for adjacent qubit pairs are
  visually represented by a $3\times 3$ matrix (see legend
  bottom). The first Pauli character corresponds to the qubit adjacent
  the highlighted tile.}
\label{Fig:Fitting}
\end{figure*}

\bigskip
\paragraph{Experimental model fitting}

To illustrate the learning protocol, we first benchmark the four-qubit
layer with two {\sc{cx}} gates shown in Figure~\ref{Fig:Layers}b on a
27-transmon-qubit, fixed-connectivity processor with a heavy-hex
topology, with qubits as indicated at the top of
Fig.~\ref{Fig:Fitting}a. For all our experiments we apply dynamical
decoupling sequences during idle times of qubits in the layer. These
idle times arise when one or more gates in the layer are significantly
faster than the slowest one, or when a qubit in the layer does not
contain a gate (see also~\sref{SVII.3}{sec:exp-setup}).
Repeated application of a noise channel in the context of self-adjoint
two-qubit Clifford gates, such as {\sc{cx}} and {\sc{cz}} gates,
generally results in pairwise products of fidelities.  Although
inserting appropriate single-qubit gates between applications can
increase the number of individual Pauli fidelities estimates, pairwise
fidelities will always remain, leading to indeterminacy of model
coefficients; for instance, we can express the pairwise fidelity
$f_af_b$ as $(\alpha f_a)(f_b/\alpha)$ for any $\alpha$. We address
this indeterminacy either through direct estimation of missing
fidelities by measuring a single layer, at the cost of an additive
error in the estimate and sensitivity to state preparation and readout
errors, or through symmetry relations that follow under the reasonable
assumption on the noise (see \sref{SV}{Sec:TwoQubitGates}).

With this in mind, we benchmark the four-qubit layer for increasing
depths up to $d$ in nine different bases in order to obtain all
necessary data. Each data point in Fig.~\ref{Fig:Fitting}a represents
an estimated observable in a given basis, averaged over 100 random
circuit instances with 256 shots each. We then fit exponentially
decaying curves through the data points corresponding to each unique
fidelity pair $f_af_b$, and augment the fidelities obtained this way
with fidelity estimates resulting from the symmetry condition. From
this, we obtain the model coefficients $\lambda$, shown in
Fig.~\ref{Fig:Fitting}b, using an adapted nonnegative least-squares
fitting procedure that uses the modified relation
$\log(f_1f_2) = -2(M_1+M_2)\lambda$ to reflect the use of pairwise
fidelities (see~\sref{SV}{Sec:TwoQubitGates}). As seen in
Fig.~\ref{Fig:Fitting}c, the fidelities of the resulting model closely
match the measured fidelities. This provides confidence that the
selected model captures the noise accurately.

To illustrate scalability of the method we used the same protocol to
learn the noise model for a 20-qubit layer involving ten concurrent
{\sc{cx}} gates. Figure~\ref{Fig:Fitting}d depicts the layer and the
resulting model coefficients.  The illustration visualizes the
sparse-model coefficients as a map over the quantum processor. We
emphasize that learning the 20-qubit noise model takes the same number
of circuit instances as that of the 4-qubit model.
\begin{figure*}[ht]
\centering
\includegraphics[width=\textwidth]{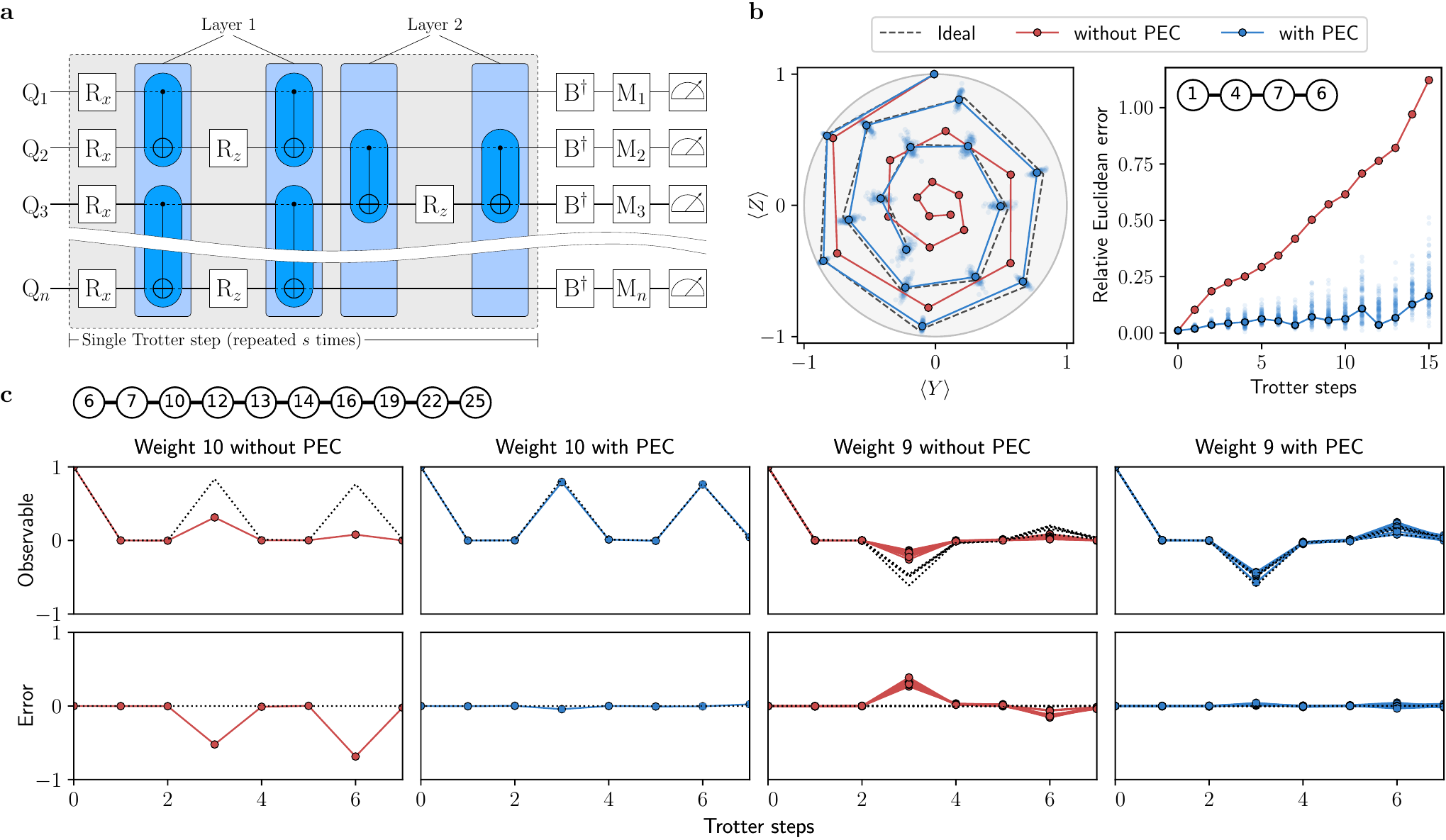}
\caption{ \textbf{Error mitigated time evolution of Ising spin
    chains.}  (\textbf{a}) Trotter circuit for the Ising Hamiltonian
  in Eq.~\eqref{eq:IsingH} over a one-dimensional $n$-qubit
  lattice. The shaded box represents a single Trotter step and is
  repeated $s$ times, with associated $R_X(2h\delta t)$ and
  $R_Z(-2J\delta t)$ rotations. Each step comprises two instances of
  two unique \textsc{cx} layers. The $B^{\dag}$ gates select the
  measurement basis and the~$M$ gates facilitate our model-free
  readout-error mitigation~\cite{BER2020MTa-arXiv}.  (\textbf{b}) Time
  evolution of the Ising model for an $n=4$ spin lattice with and
  without probabilistic error correction (PEC) for 15 Trotter steps;
  $h=1$, $J=0.15$, and $\delta t=1/4$.  Left: Trotterized
  time-evolution of the global
  magnetization~$\vec M := \sum_n \left( \left\langle X_n
    \right\rangle, \left\langle Y_n \right\rangle, \left\langle Z_n
    \right\rangle \right)/N$ shown in the Y-Z Bloch plane. The
  experimentally measured evolution (dots and solid lines) is compared
  to the ideal noise-free one (dashed lines).  The bootstrap-estimated
  error distribution for each data point is shown as clouds (light
  dots).  Right: The error between ideal and measured magnetization
  vectors, in terms of the relative Euclidean distance
  $\Vert\vec M - \vec M_\mathrm{ideal}\Vert_2 / \Vert\vec
  M_\mathrm{ideal}\Vert_2$.  (\textbf{c}) Time evolution of the Ising
  model on a one-dimensional ten-qubit lattice sites (top) with $h=1$,
  $J=0.5236$, and $\delta t=1/4$. All weight-10 (left) and weight-9
  Pauli-Z observables (right) are plotted along with the ground truth
  (dashed).  }\label{Fig:IsingCircuit}
\end{figure*}

\paragraph*{Probabilistic error cancellation}
Once the noise model has been learned, it can be used to mitigate the
noise using the PEC method~\cite{TEM2017BGa}. The protocol implements
the channel inverse $\Lambda_i^{-1}$ through quasi-probabilistic
sampling for each of the $l$ layers. The inverse of the map $\Lambda$
is obtained by negating $\mathcal{L}$, leading to a non-physical map
given by
\begin{equation}
\Lambda^{-1}(\rho) = \mbox{exp}[-\mathcal L](\rho) = \gamma  \prod_{k \in \mathcal{K}}\left( w_k\cdot - (1-w_k)P_k\cdot P_k^{\dag}\right)\rho,
\label{eq:InverseNoiseChanLindblad}
\end{equation}
with {\it sampling overhead}
$\gamma = \exp(\sum_{k\in {\cal K}} 2\lambda_k)$.  This amounts
exactly to inverting each individual factor in
Eq.~\eqref{Eq:MainFinalForm} due to commutativity of the factors.  The
product structure allows for a direct way of sampling the map. For
each $k \in {\mathcal K}$ we sample the identity with probability
$w_k$ or apply the Pauli $P_k$ otherwise. We record the number of
times $m$ we have applied a non-identity Pauli, compute a final Pauli
as the product of all sampled terms. Repeating this for each noise
channel $i=1,\ldots, l$ with respective $m_i$ and $\gamma_i$ values,
we construct a circuit instance in which each noisy layer is preceded
with the corresponding sampled Pauli. The measurement outcome of the
circuit is then multiplied by $\prod_{i=1}^l (-1)^{m_i} \gamma_i$. On
average, this implements the inverse maps and produces an unbiased
expectation value with sampling overhead
$\gamma(l) = \prod_{i=1}^l \gamma_i$ In
\sref{SVI.2}{Sec:ErrorAnalysisObservables}, we derive an error bound
on the final expectation value that considers the errors in all steps
of the procedure. The bound states that, given a quantum circuit with
$l$ layers whose learning layer satisfies Eq.~\eqref{Eq:MainBound1},
we can estimate the ideal expectation value $\langle A\rangle$ of an
observable $A$ with $\norm{A} \leq 1$ by the average mitigated
estimate $\langle \hat{A}_N\rangle$ using $N$ error-mitigated circuit
instances, such that
\[
\vert\langle{A}\rangle - \langle \hat{A}_N\rangle\vert
\leq 
(C_{\epsilon}^{l \tau} - 1) + \gamma(l) \sqrt{2\log(2/\delta')/N}
\]
is satisfied with probability at least $1-\delta'$. For modest noise,
$C_{\epsilon}$ can be expected to be close to one, which leads to a
scaling that is only weakly exponential in $l$ and~$\tau$. The
sampling overhead $\gamma(l)$ dictates the resources needed to obtain
a reliable estimator~\cite{TEM2017BGa}.

\paragraph*{Quantum simulation of the Ising model}
As a practical application for noise mitigation with our proposed
noise model we consider time evolution of the one-dimensional
transverse-field Ising model due to the Hamiltonian
\begin{equation}\label{eq:IsingH}
H = -J\sum_{j} Z_{j}Z_{j+1} + h\sum_{j} X_{j}  = -JH_{ZZ} + hH_{X},
\end{equation}
where $J$ denotes the exchange coupling between neighboring spins and
$h$ represents the transverse magnetic field. Unitary time
evolution~$e^{-iHt}$ can be approximated by a first-order Trotter
decomposition $\big(e^{iJH_{ZZ}t/s}e^{-ihH_{X}t/s}\big)^s$ with $s$
segments.  We perform the time evolution on a linear chain of qubits,
where we implement the unitary $\exp(iJ(Z_{j}Z_{j+1})\delta_t)$ with
$\delta_t = t/s$ as a quantum circuit consisting of an
$R_Z(-2J\delta_t)$ rotation on qubit $j\!+\!1$ between two {\sc{cx}}
gates with control and target qubits $j$ and $j\!+\!1$. Similarly,
$\exp(-ihH_x\delta_t)$ decomposes into a product of single-qubit
rotations $R_X(2h\delta_t)$ on each qubit $j$ (for more details
see~\cite{KIM2021WYMa-arXiv}). This results in circuits of the form
shown in Fig.~\ref{Fig:IsingCircuit}a.  The circuit contains two
unique layers of {\sc{cx}} gates, one starting at even and one at odd
locations in the qubit chain. Once the noise models for the two layers
are learned, we generate random circuit instances. We apply
readout-error mitigation on all
observables~(see~\cite{BER2020MTa-arXiv} for more on readout
mitigation). To counter time-dependent fluctuations in the noise we
relearn the noise model after fixed intervals (see
also~\sref{SVII}{sec:exp-setup}). The final observables are obtained
after averaging.

As a first experiment, we consider the Ising-model dynamics for a spin
chain with four sites with $h=1$ and $J=0.15$. Learning of the first
layer was detailed in Fig.~\ref{Fig:Fitting}a--c and resulted in
factor $\gamma = 1.03$. All other models were learned in a similar
fashion.  The number of mitigated circuit instances for each
$s = 1,2,\ldots,15$ is given by
$\min(200, 40\cdot(\gamma_1\gamma_2)^{2s})$, where $\gamma_1$ and
$\gamma_2$ are the sampling overhead factors for the first and second
layer. Each circuit instance is measured 1,024 times.

For each of the $s$ Trotter-steps we compute the global magnetization
component $\langle Z\rangle_s$ as the overall average of all
weight-one Pauli-Z observables, and likewise for $\langle X\rangle_s$
and $\langle Y\rangle_s$. The resulting $Y$ and $Z$ magnetization
components are plotted in Fig.~\ref{Fig:IsingCircuit}b (left) along
with the results obtained without PEC and exact simulation. We compare
the relative Euclidean distance for the estimated and exact global
magnetization in Fig.~\ref{Fig:IsingCircuit}b (right).

Our second experiment considers the simulation of a one-dimensional
lattice on ten qubits with $h=1$ and $J=0.5236$ for up to seven
Trotter steps. High-weight observables are highly noise sensitive and
serve as a demanding test of the method. In
Fig.~\ref{Fig:IsingCircuit}c, we compare the results for weight-9 and
-10 Pauli-Z observables obtained with and without PEC.  Mitigated
observables exhibit vanishing residuals.

\begin{figure}[ht]
\centering
\includegraphics[width=0.8\linewidth]{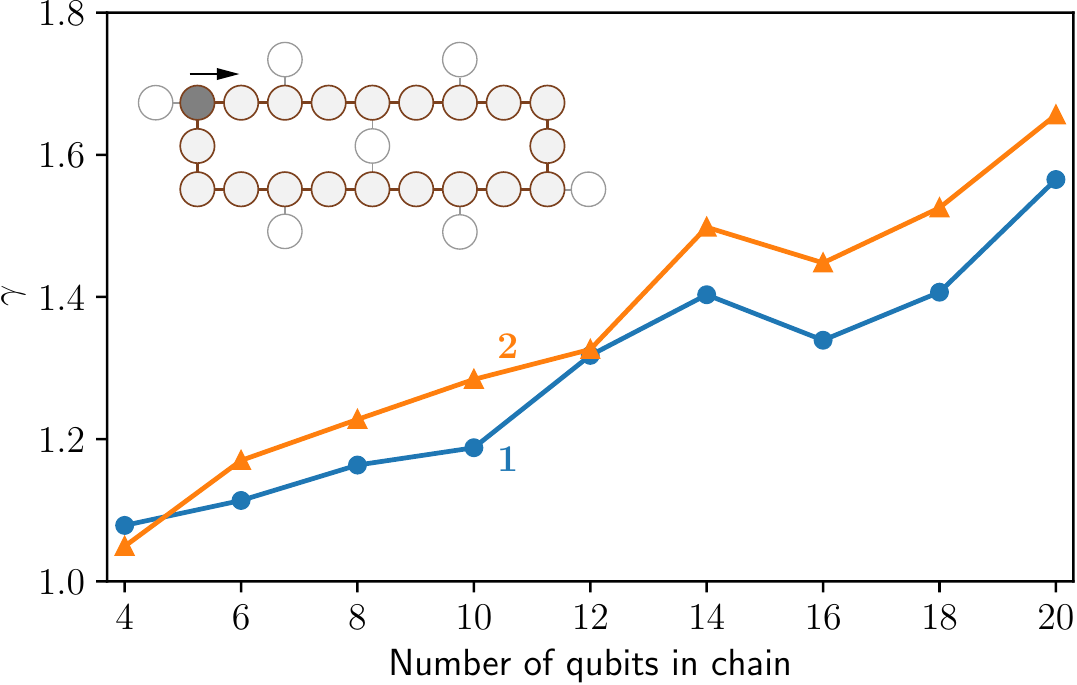}
\caption{ \textbf{Mitigation sampling overhead.}\label{Fig:Gamma}
  Sampling overhead~$\gamma$ for the two Ising layers as a function of
  the number of qubits in the Ising lattice chain. The chain is
  depicted as an inset at the top. The first qubit in the chain is at
  the top-left (dark shading) and the chain proceeds clockwise
  (arrow). Layers 1 and 2 have complementary \textsc{cx} gates on
  alternating pairs of qubits in the chain (see
  Fig.~\ref{Fig:IsingCircuit}).  }\label{fig:mitigation-budget}
\end{figure}

\paragraph*{Discussion and conclusions}
The remarkable accuracy of the error-mitigated observables in Fig.~3
provides strong evidence for the validity of our sparse noise model
and learning protocol. It is nonetheless important to discuss
potential limitations of our method, such as the sampling overhead.
In particular, the variance in the estimator scales with the square of
the sampling overhead factor $\gamma(l)$. This factor depends on the
number of qubits (see Fig.~\ref{Fig:Gamma}) as well as the circuit
depth in terms of the number of layers. We can define a qubit- and
depth-normalized version of the scaling factor, $\bar{\gamma}$, which
allows us to conveniently express the sampling overhead for $l$ layers
on $n$ qubits as $\bar{\gamma}^{nl}$. This normalized parameter itself
can also be used as a metric to represent quantum processor
performance; improvements in the hardware quality are reflected in
lower $\bar{\gamma}$ values, which in turn translate into potentially
dramatic decreases in the sampling overhead (see also
\sref{SVI.3}{Sec:Scaling}).  Our work serves as a powerful example of
how classical run-time overheads can be traded for tremendously
improved quantum computation on noisy processors. However, this also
highlights the importance of improving total circuit execution
time~\cite{wack2021}, which will reduce the practical PEC overhead.

In conclusion, our results demonstrate for the first time a practical
path to extend probabilistic error cancellation to remove the
noise-induced bias from from high-weight observable across the full
circuit (see Fig~\ref{Fig:IsingCircuit}c). This is made possible by
our sparse learning protocol, which provides a versatile noise
representation with rigorous theoretical bounds and near-constant
learning with number of qubits, and an effective noise-inversion
scheme. The accuracy of the model-reconstructed noise-fidelity pairs,
as shown in Fig~\ref{Fig:Fitting}c, and our error mitigated
observables validate the view that the Lindbladian learning is
accurate, efficient, and scalable. We expect our learning protocol to
be a powerful characterization and benchmarking tool, and more-broadly
to enable the study and mitigation of noise in quantum processors at a
new scale.

\paragraph*{Acknowledgments}
The authors thank Sergey Bravyi, Douglas T. McClure, and Jay
M. Gambetta for helpful discussions. Research in characterization and
noise learning was sponsored in part by the Army Research Office and
was accomplished under Grant Number W911NF-21-1-0002. The views and
conclusions contained in this document are those of the authors and
should not be interpreted as representing the official policies,
either expressed or implied, of the Army Research Office or the
U.S. Government. The U.S. Government is authorized to reproduce and
distribute reprints for Government purposes notwithstanding any
copyright notation herein.

\paragraph*{Data availability}
Data are available from the authors on reasonable request.

\bibliographystyle{naturemag}
\bibliography{bibliography,library}

\clearpage
\newpage
\widetext

\begin{center}
  \textbf{\large Supplementary Information:\\ \vspace{0.5cm}
    Probabilistic error cancellation with sparse Pauli-Lindblad
    models\\ on noisy quantum processors}
\end{center}

\setcounter{equation}{0}
\setcounter{figure}{0}
\setcounter{table}{0}
\setcounter{page}{1}
\setcounter{section}{0}
\makeatletter
\renewcommand{\theequation}{S\arabic{equation}}
\renewcommand{\thefigure}{S\arabic{figure}}
\renewcommand{\thesection}{S\Roman{section}}

\section{Summary of the method}\label{Sec:MethodSummary}

\fbox{\begin{minipage}{0.98\linewidth}
\raggedright
Input: the layer's qubits and gates, and processor topology

\bigskip
{\bf{Model definition}}\\
\begin{itemize}[label=\textbullet,leftmargin=*,partopsep=-0.2em]
\item Using the qubit and topology information, define model Paulis $\mathcal{K}$. This set contains all weight-one Paulis supported on the model qubits as well as all weight-two Paulis supported on selected pairs of connected qubits
\end{itemize}

\medskip
{\bf{Preparation for model fitting}}\\
\begin{itemize}[label=\textbullet,leftmargin=*,partopsep=-0.2em]
\item Define the measurement bases and determine the fidelities $\mathcal{B}$ needed to fit the model
(Section~\ref{Sec:Bases})\\
\end{itemize}

\medskip
{\bf{Fidelity estimation}}\\
\begin{itemize}[label=\textbullet,leftmargin=*,partopsep=-0.2em,itemsep=-0.15em]
\item For each basis, run benchmark circuits at different depths
\item Fit data with exponentially decaying curves to estimate
individual fidelities or fidelity-pair products
\item Complete the fidelities using unit-depth benchmark circuits or symmetry assumptions
\item Form vector $\hat{f}$ of estimated fidelities
\end{itemize}

\medskip
{\bf{Model fitting}}\\
\begin{itemize}[label=\textbullet,leftmargin=*,partopsep=-0.2em,itemsep=-0.15em]
\item Form matrix $M = \mathcal{M}(\mathcal{B},\mathcal{K})$ (see Eq.~\eqref{Eq:MFT} in Section~\ref{subsec:learning})
\item Set model parameters to the solution of the following problem (see Eq.~\eqref{Eq:NNLS} in Section~\ref{subsec:learning})
\[
\mathop{\mathrm{minimize}}_{\lambda\geq 0}
\quad {\textstyle{\frac{1}{2}}}
\Vert M\lambda + \log(\hat{f})/2\Vert_2^2
\]
\end{itemize}

\smallskip
{\bf{Mitigation}}\\
\begin{itemize}[label=\textbullet,leftmargin=*,partopsep=-0.2em,itemsep=-0.15em]
\item Given a circuit that contains the layer of gates
\item Generate multiple circuit instances with each layer preceded by a Pauli sampled from the quasi-probability distribution and with a Pauli twirled instance of the layer
\item Estimate the expectation of the observables of interest and scale by $\gamma$
\end{itemize}

\medskip
Note: for layers of two-qubit gates there are two lists of fidelity terms $\mathcal{B}_1$ and $\mathcal{B}_2$. In this case, we replace $M$ by $\mathcal{M}(\mathcal{B}_1,\mathcal{K}) + \mathcal{M}(\mathcal{B}_2,\mathcal{K})$. The elements in vector $\hat{f}$ then represents products of two fidelities. See Section~\ref{Sec:FullRankM2} for more details.
\end{minipage}}

\section{Background and review}\label{Sec:Preliminaries}

Most quantum applications combine classical computing with the
execution of one or more sets of quantum circuits on the quantum
processor. Each circuit execution can roughly be thought of as
consisting of three phases: (i) initialization of the quantum
processor to the $\ket{0}$ ground state; (ii) application of the gates
that make up the quantum circuit; and (iii) measurement of the qubits
of interest. For each circuit, this process is repeated multiple times
to obtain the desired measurement statistics.  The process of running
a quantum circuit is affected by different sources of noise. The noise
associated with the first and last stage is usually combined into
so-called state-preparation and measurement (SPAM) error. There are
quite a few algorithms for dealing with this type of noise, see for
instance~\cite{BRA2021SKMa,CHE2020YZFa-arXiv,BER2020MTa-arXiv} for
different algorithms and further references. Noise in the second stage
consists of global background noise, such as dephasing and
decoherence, and noise associated with the application of one or more
gates, including cross-talk. Here, we focus on the noise associated
with the application of a single operation on one or more qubits. It
often helps to write a noisy operation $\tilde{\cal U}$ as a
combination of a noise channel $\tilde{\Lambda}$ and the ideal
operation ${\cal U}(\rho) = U \rho U^\dagger$:
\[
\tilde{\cal U} = {\cal U} \circ \tilde{\Lambda}.
\]
In the remainder of this section we look at techniques for shaping
general noise channels $\tilde{\Lambda}$ into more structured and
therefore more manageable channels, as well as ways of inverting these
new channels.

\subsection{Noise channel simplification} \label{Sec:Twirling} There
are many ways to characterize or represent noise channels. Suppose
that $\tilde{\Lambda}$ is a noise channel that applies to $n$-qubits
and denote by $\{P_i\}_{i=0}^{4^{n}-1}$ the Pauli basis for the
corresponding Hilbert space. Then we can express $\tilde{\Lambda}$ in
terms of the Pauli transfer matrix $\mathcal{T}_{\tilde{\Lambda}}$
with entries
\[
T_{\tilde{\Lambda}}[a,b] = \frac{1}{2^n}\Tr\left[P_a^{\dag}(\tilde{\Lambda}(P_b))\right].
\]
In general, this will be a dense matrix, and working with the explicit
form with $\mathcal{O}(4^{2n})$ nonzero coefficients therefore quickly
becomes intractable, certainly because all coefficients need to be
estimated in tomography. However, it has been shown
\cite{knill2004fault,kern2005quantum,geller2013efficient,wallman2016noise}
that conjugation of the noise channel with randomly sampled operators
from the Pauli group results in an averaged channel
\begin{equation}\label{Eq:PauliTwirl}
\Lambda(\cdot) = \mathbb{E}_{i}\big[ P_i^{\dag}\tilde{\Lambda}(P_i^\dagger \cdot P_i)P_i^\dagger\big],
\end{equation}
with a diagonal transfer matrix
$T_{\Lambda}[a,b] =\delta_{a,b}T_{\tilde{\Lambda}}[a,b]$.  The
averaging operation in~\eqref{Eq:PauliTwirl} is called a Pauli
twirl. In addition to a much more compact representation, this
transfer matrix is also easily inverted. The quantities on the
diagonal of the transfer matrix represent the Pauli fidelities
$f_a = 2^{-n} \Tr\left[P_a (\tilde{\Lambda}(P_a))\right]$. We can use
the symplectic Walsh-Hadamard transformation to convert these
fidelities \cite{FLA2020Wa} to coefficients
\begin{equation}\label{Eq:WalshHadamard}
c_b = \frac{1}{2^n}\sum_{a} (-1)^{\langle a,b\rangle_{sp}}f_a,
\end{equation}
where $\langle a,b\rangle_{sp}$ denotes the symplectic inner product
of Paulis $P_a$ and $P_b$, which is zero if the Paulis commute (that
is $[P_a,P_b] = P_aP_b - P_bP_a = 0$), and one otherwise. These
coefficients allow us to then rewrite the noise operator applied to
the density matrix $\rho$ as a Pauli channel:
\begin{equation}\label{Eq:PauliChannel}
\Lambda (\rho) = \sum_{i} c_i P_i \rho P_i^{\dag},
\end{equation}
where the vector $c = [c_i]$ of all coefficients represents a
distribution: $c_i \geq 0$ and $\sum_i c_i = 1$. The Pauli twirl can
be approximated by generating multiple instances of the appropriate
quantum circuit, each with a Pauli term $P_a$ sampled uniformly at
random from the $n$-Pauli matrices. This may seem difficult, since, in
general, we are not given an isolated noise channel but rather have
access only to a noisy gate
$\tilde{\cal U} = {\cal U} \circ \tilde{\Lambda}$. In this case we
just want to twirl $\tilde{\Lambda}$, which is possible by pushing the
Pauli through the $U$ gate, when this is a Clifford gate
\cite{knill2004fault,kern2005quantum,geller2013efficient,wallman2016noise}. To
see how this works, observe that
\[
U P_a^{\dag} \tilde{\Lambda}(P_a \cdot P_a^{\dag}) P_a U^\dag  = 
U P_a^{\dag} (U^{\dag} U) \tilde{\Lambda}(P_a \cdot P_a^{\dag}) (U^{\dag}U) P_a U^\dag =
P_{a^{U}} \; \tilde{\cal U}(P_a \cdot P_a^{\dag}) P_{a^{U}}^\dag.
\]
When $U$ is a Clifford operator, it is well known that the conjugation
of one Pauli operator results in another Pauli, namely
$U P_a U^{\dag} = P_{a^{U}}$. That means that Pauli twirling for
Clifford operators $U$ can be conveniently implemented by sampling a
random $P_a$ term and applying this terms and its conjugate under $U$
to the circuit, around the noisy gate to get $P_{a^{U}}
\tilde{U}P_a$. Pauli operators themselves are formed as the direct
product of Pauli matrices $X$, $Y$, and $Z$ and the two-by-two
identity matrix, and $n$-Paulis can therefore be efficiently
represented by a string $\{I,X,Y,Z\}^n$ of length $n$, or in
symplectic form as a binary vector of length $2n$. The latter
representation enables a computationally efficient way of conjugating
the Pauli operator by any Clifford operator~\cite{AAR2004Ga}, and
therefore allows us to efficiently find $P_{a^{U}}$ for a given
$P_a$. Since Pauli operators can be implemented using single-qubit
gates, we can often simplify the circuits of twirled gates. Any
single-qubit directly preceding or following the gate can be combined
with the respective single-qubit gate of operators $P_a$ or
$P_{a^U}$. This can reduce or even completely eliminate the circuit
overhead of the Pauli twirl.

Twirling is possible over groups
\cite{wallman2016noise,cai2019constructing} other than the Pauli
group. In general, given a group $\mathcal{G}$, we can define the
twirled noise channel
\[
\Lambda_{\mathcal{G}} :=
\frac{1}{\vert\mathcal{G}\vert}\sum_{G\in\mathcal{G}} G^{-1} \circ \tilde{\Lambda} \circ G.
\]
When $\mathcal{G}$ is the Clifford group, or any other two-design, the
resulting transfer matrix $T_{\Lambda_{\mathcal{G}}}$ is not only
diagonal, but such that the fidelities for Pauli operators other than
the identity (which is always one) are all equal. This means that the
new twirled noise channel can be described by only a single
parameter. In case both $U$ and $G$ are elements of the Clifford
group, it holds that the conjugated operator $U G U^{\dag}$ remains an
element of the Clifford group. As in the Pauli case, one can
efficiently represent elements from the Clifford group and compute the
conjugation. However, the problem is that the circuit implementation
of a Clifford gate can have a significant depth~\cite{BRA2020Ma}, and
may therefore introduce an unacceptable amount of noise itself.

\subsection{Quasi-probabilistic noise inversion}\label{Sec:QPNoiseInversion}
The probabilistic error cancellation method as given
in~\cite{TEM2017BGa,END2018BLa} asks that for the general procedure an
ideal ${\cal U}$ operation is expanded into a set of noisy operators
$\{\tilde{\cal U}_i\}_{i}$ that can be implemented on the quantum
hardware. However, we are in the particular situation that our noisy
operations for each layer are exactly of the form $\tilde{\cal U} =
{\cal U} \circ \Lambda$, where $\Lambda$ is a Pauli channel. In
particular, as explained in the previous section, the general
procedure considered in \cite{TEM2017BGa,END2018BLa} can be reduced to
this special case after Pauli twirls have been applied. We are then in
the setting where it is sufficient to only focus on the noise of the
Pauli noise channel $\Lambda$ and implement its inverse $\Lambda^{-1}$
in experiment. \\

When represented as the diagonal Pauli transfer matrix $T_{\Lambda}$
it is clear that the inverse should have a Pauli transfer matrix given
by $T_{\Lambda}^{-1} = \mbox{diag}(f_a^{-1})$. That is, a diagonal
matrix with the inverse fidelities on the diagonal. If then follows
from the Walsh-Hadamard transform in~\eqref{Eq:WalshHadamard}, that we
would like to have a Pauli channel with coefficients
\[
c_b^{\mathrm{inv}} = \frac{1}{2^n}\sum_{a} (-1)^{\langle a,b\rangle_{sp}}\frac{1}{f_a}.
\]
However, except for the case where all fidelities are one, the
resulting coefficients will contain negative values, and therefore
does not represent a physical Pauli channel. The method proposed
in~\cite{TEM2017BGa} addresses this as follows. We can first rewrite
the desired channel as
\[
\sum_{i}c_i^{\mathrm{inv}}P_i\rho P_i^{\dag}
=
\sum_{i}\mathrm{sgn}(c_i^{\mathrm{inv}})\vert c_i^{\mathrm{inv}}\vert
  P_i\rho P_i^{\dag}
=
\gamma
\sum_{i}\mathrm{sgn}(c_i^{\mathrm{inv}})\gamma^{-1}\vert c_i^{\mathrm{inv}}\vert
  P_i\rho P_i^{\dag}\;,
\]
where $\mathrm{sgn}$ denotes the signum function and
$\gamma := \sum_i \vert c_i^{\mathrm{inv}}\vert$. The transformed
coefficients
$\hat{c}_i^{\mathrm{inv}} = \gamma^{-1}\vert c_i^{\mathrm{inv}}\vert$
are clearly nonnegative and by definition of $\gamma$, sum up to one,
and therefore represent a distribution. In order to implement noise
inversion the algorithm proceeds as follows. First, a random Pauli $a$
is sampled according to the distribution
$\hat{c}_i^{\mathrm{inv}}$. We store the sign
$\mathrm{sgn}(c_a^{\mathrm{inv}})$ and form a circuit with that
includes the sampled $P_a$ prior to the noise channel. We then
estimate the expectation values of any desirable observable and scale
it by the sign as well as by $\gamma$. When computed over multiple
random samples $P_a$, the empirical mean value of the scaled
observables then provide an unbiased estimator of the ideal
expectation value that would result from a noiseless circuit. The cost
of sampling from the quasi-probabilistic distribution is an increase
in variance in the expected value by a factor of
$\mathcal{O}(\gamma^2)$.

\subsection{Scalable noise models}
While working with explicit Pauli channels is convenient, they do
require the storage and processing of $4^n$ coefficients for $n$
qubits in general. In order to reduce the model complexity and
maintain efficiency, the work presented in~\cite{FLA2020Wa} considers
Pauli channels with bounded degree correlations. The probability
distribution representing the Pauli channel in this case is factored
based on the individual terms and such that certain terms are
conditionally independent. The resulting probabilities are in Gibbs
form and can be reconstructed from locally measured patches at the
expense of computing the full partition function of the
distribution. While the resulting structure can help reduce the noise
channel representation, application of the model to noise mitigation
and computing or sampling from the noise inverse remains
challenging. We therefore focus on a Pauli model that retains the
local correlation but is better suited to the probabilistic
error-cancellation protocol.

\section{Pauli-Lindblad noise model}\label{Sec:ErrorModel}
We propose the use of a locally correlated noise model that is
motivated by the continuous-time Markovian dynamics of open quantum
systems. These dynamics can be described by a  quantum master
equation. When appropriately rewritten in diagonal form, this can be
expressed as the Lindblad equation~\cite{BRE2002Pa}
$\frac{d}{dt}\rho(t) = \mathcal{L}\rho(t)$, where $\mathcal{L}(\rho) =
-i[H,\rho] + \sum_{k} \left(A_k\rho A_k^{\dag} -
  \frac{1}{2}A_k^{\dag}A_k\rho - \frac{1}{2}\rho
  A_k^{\dag}A_k\right)$. For a general Lindbladian, the unitary part
of the dynamics is described by the Hamiltonian $H$, and $A_k$ are the
Lindblad operators. The resulting channel after evolution time $t$ is
then the formal exponential $T_t = \mbox{exp}\left[{\cal
    L}t\right]$.\\ 

The Pauli-Lindblad noise model we consider contains no internal
Hamiltonian dynamics and we therefore do not consider a Hamiltonian
contribution. We want to generate a Pauli channel and therefore take
$A_k = \sqrt{\lambda_k} P_k$ for a set of Pauli operators
$\{P_k\}_{k \in {\cal K}}$ we enumerate with an index set ${\cal K}$:
\begin{equation}\label{Eq:sparse-Lindbladian}
\mathcal{L}(\rho) = \sum_{k \in {\cal K}}\lambda_k\left(P_k \rho P_k- \rho\right).
\end{equation}
In particular we assume that $|{\cal K}| \ll 4^n-1$ is a set that is
only of polynomial size in the number of qubits. This means the model
is determined by a set of non-negative numbers $\lambda_k \geq 0$ for
$k \in {\cal K}$. We will discuss the choice of this set for our
experimentally considered set up in section
\ref{Sec:SparseNoiseModel}. In general, the set can be chosen as to
account for the correlations that are present in the quantum hardware
of interest. Since we are only interested in a particular noise model,
we set the dynamics to be at time $t=1$ and directly define the sparse
Pauli-noise model as $\Lambda = \mbox{exp}\left[{\cal L}\right]$. \\

When working in the matrix representation expressing the Lindbladian
$\mathcal{L}$, which we denote by
$\mbox{vec}\left[\mathcal{L}\right]$, it follows that the sparse noise
model $\Lambda$ is given by the conventional matrix exponential
\begin{equation}
    \mbox{vec}\left[\Lambda \right] = e^{\mbox{vec}\left[\mathcal{L}\right]}.
\end{equation}
Here, the matrix representation of the Lindbladian is defined as
\begin{equation}\label{Eq:SuperLindbladian}
\mbox{vec}\left[\mathcal{L}\right] = \sum_{k \in {\cal K}}\lambda_k\left(P_k \otimes P_k^T
  - I\otimes I\right).
\end{equation}
Note that for any two Pauli operators $P$ and $Q$ it holds that
\[
(P\otimes P^T)(Q\otimes Q^T) = (PQ\otimes(QP)^T) = ((\pm PQ)\otimes
(\pm (PQ)^T) = (QP\otimes (PQ)^T) = (Q\otimes Q^T)(P\otimes P^T).
\]
This shows that the terms in~\eqref{Eq:SuperLindbladian} commute, and
also expresses the fact that Pauli channels commute. Given the
commutativity of the terms, we can write the time-evolution operator
as
\begin{equation}\label{Eq:SuperLindbladian2}
\mbox{vec}\left[\Lambda \right] = \prod_{k}e^{-\lambda_k}e^{\lambda_k P_k \otimes P_k^T}
\end{equation}
Exponentiation with a Pauli operator can we written as
\begin{align}
e^{\lambda (P\otimes P^T)}
& = \cosh(\lambda)(I\otimes I) + \sinh(\lambda)(P\otimes P^T)\notag\\
& = \frac{e^{\lambda} + e^{-\lambda}}{2}(I\otimes I) + 
\frac{e^{\lambda} - e^{-\lambda}}{2}(P\otimes P^T)\label{Eq:ExpGammaP}
\end{align}
Combining~\eqref{Eq:SuperLindbladian2} and~\eqref{Eq:ExpGammaP} and we
obtain the final form of the noise model as
\begin{equation}\label{Eq:FinalForm}
\Lambda(\rho) = \prod_{k}\left(w_k\cdot +
  (1-w_k)P_k\cdot P_k^{\dag}
\right)\rho,
\end{equation}
where $w_k = (1+e^{-2\lambda_k})/2$. Given the time evolution of
states in~\eqref{Eq:FinalForm}, it is natural to ask what effect it
has on Pauli operators. The fidelity $f_a$ of a Pauli operator $P_a$
can be expressed as
\begin{align}
f_a
&=
\frac{1}{2^n}\Tr\left[P_{a}^{\dag}\Lambda(P_a)\right]
=
\frac{1}{2^n}\Tr\Big[P_{a}^{\dag}P_a\prod_{\{a,k\}=0}(2w_k-1)\Big]\notag\\
&= \prod_{\langle a,k\rangle_{sp} = 1}(2w_{k}-1)
= \prod_{\langle a,k\rangle_{sp} = 1 }e^{-2\lambda_k}
=\exp\left(-2\sum_{k\in\mathcal{K}}
  \lambda_{k}\langle a,k\rangle_{sp}\right).
\label{Eq:FidelityFromGamma}
\end{align}
We can define a matrix $M$ with entries
$M_{a,b} = \langle a,b\rangle_{sp}$, such that $M_{a,b}=0$ if Paulis
$P_a$ and $P_b$ commute, and $M_{a,b} = 1$ otherwise. Denoting by $f$
and $\lambda$ the full vector of Pauli fidelities and model
coefficients, respectively, we can compactly
express~\eqref{Eq:FidelityFromGamma} as
\begin{equation}\label{Eq:Mf}
-\log(f)/2 = M\lambda\;,
\end{equation}
where the logarithm is applied elementwise. Finally, we observe that
the coefficients $w_k$ in~\eqref{Eq:FinalForm} are all nonnegative,
and that the fidelity for the identity operator is always one, since
all Pauli terms commute with the identity. It follows
that~\eqref{Eq:FinalForm} is a valid Pauli channel for all
$\lambda \geq 0$.

\subsection{Channel operations}\label{Sec:ChannelOperations}

The Lindbladian noise channel in~\eqref{Eq:FinalForm} has some useful
properties. First, changing the evolution time amounts to scaling
$\lambda$. Second, given two separate noise channels with parameters
$\lambda_1$ and $\lambda_2$, it follows from multiplicativity of
fidelities under successive Pauli channels that
\[
-\log(f_1f_2)/2 = -\log(f_1)/2 -\log(f_2)/2 = M\lambda_1 + M\lambda_2 = M(\lambda_1+\lambda_2)\;, 
\]
which shows that combination of channels amounts to addition of the
coefficients. The inverse of a channel is characterized by inverse
fidelities, and it directly follows from
\[
-\log(1/f)/2 = \log(f)/2 = -M\lambda = M(-\lambda)
\]
that the inverse noise model is obtained by simply negating the
coefficients.

\subsection{Sparse models}\label{Sec:SparseNoiseModel}
Quantum circuits are generally transpiled into native single- and
two-qubit gates applied to individual qubits or pairs of qubits that
are topologically connected, that is, neighboring qubits. The noise
associated with the application of these gates can be expected to have
limited range and therefore be negligible beyond some local
neighborhood around the qubits to which the operation is applied. This
suggests it may not be necessary to include all possible Pauli terms
in~\eqref{Eq:FinalForm}, and motivates us to simplify the model and
include only a select subset of Pauli terms $P_k$. For instance, we
could include those Paulis that contain only a single non-identity
term, or two such terms on neighboring qubits. Such sparse models can
be represented far more efficiently than their full counterpart. For a
linear topology of $n$ qubits, the number of coefficients $\lambda$
reduces from $4^n-1$ to a mere $3n + 9(n-1)$, which is clearly far
more scalable in terms of the number of qubits.

\subsection{Learning the model}\label{subsec:learning}
In order to characterize a noise channel we need to find model
coefficients that best explain the experimental data. For the proposed
noise model, a practical way of determining the model coefficients
follows directly from equation~\eqref{Eq:Mf}. We first form a vector
$f$ of fidelity for Pauli terms in some list $B$. Given model Paulis
$K$ we can form the matrix
\begin{equation}\label{Eq:MFT}
M = \mathcal{M}(B,K)\quad\mbox{such that}\quad
M_{i,j} = \begin{cases} 0 & [B_i,K_j] = 0\\
1 & [B_i,K_j] \neq 0.
\end{cases}
\end{equation}
We then find nonnegative coefficients $\lambda$ such that $M\lambda$
is as close to $-\log(f)$ as possible. When measuring in Euclidean
distance (other norms could be used here as well), this can be
conveniently formulated as a nonnegative least-squares problem:
\begin{equation}\label{Eq:NNLS}
\lambda(f) := \mathop{\mathrm{argmin}}_{\lambda \geq 0}\quad \half\Vert
M\lambda + \log(f)/2\Vert_2^2\;.
\end{equation}
The columns in matrix $M$ correspond to the Pauli terms included in
the model, denoted by $\mathcal{K}$, whereas the rows could be any of
the $4^n$ Pauli operators, although we generally omit the row for the
identity operator since all its entries as well as the log fidelity
are zero. In case of the sparse noise model described in
Section~\ref{Sec:SparseNoiseModel}, the number of model Paulis is
relatively small and matrix $M$ will have far more rows than
columns. The model coefficients are well defined if the solution
of~\eqref{Eq:NNLS} is unique, which is guaranteed whenever
$M(\mathcal{B},\mathcal{K})$ has full column rank.

\subsection{Variance in mitigated observable}

We now consider the variance in the error-mitigated
observable. Starting with binomial distribution with $p(1) = p$ and
$p(0) = q = 1-p$ and $n$ trials we have mean $np$ and variance
$npq$. For the estimation of observables we sample from $\pm 1$ which
means scaling by two and subtracting one per trial, which gives mean
$2np-n$ and variance $2^2npq$. In order to obtain the observable we
divide by the number of trials and scale by $\gamma$, which leads to
an updated mean of $2p\gamma - \gamma = (2p-1)\gamma$ and variance
$(\gamma/n)^2(4npq) = (4\gamma^2/n)pq$. The ideal observable or
fidelity $f$ is equal to the mean, namely $(2p-1)\gamma =
f$. Rewriting gives
\[
p = \frac{1}{2}\left(\frac{f}{\gamma}+1\right),\qquad
q = 1-p = \frac{1}{2}\left(1 - \frac{f}{\gamma}\right),
\qquad
pq = \frac{1}{4}\left(1 - \frac{f^2}{\gamma^2}\right).
\]
Using this we obtain variance
\[
\frac{4\gamma^2}{n}pq = \frac{1}{n}\left(\gamma^2-f^2\right).
\]
In order to keep the variance of the estimator fixed we therefore need
to scale $n$ proportional to $\gamma^2$.

\section{Noise learning for single-qubit gates with crosstalk}\label{Sec:Benchmarking}
The proposed noise model readily applies to benchmarking and
mitigating the noise in layers of single-qubit gates. A common
assumption in this setting is that, for a given qubit, the noise is
independent of the gate that is applied.  Here we refine this and
assume that the noise channel associated with a layer of single-qubit
operations depends only on the particular subset of qubits that
contain a gate. The motivation for this is that application of a gate
to a qubit can result in crosstalk, which depends in part on qubit
connectivity as well as the presence or absence of gates on
neighboring qubits.
\begin{figure}[t]
\centering
\includegraphics[width=0.78\textwidth]{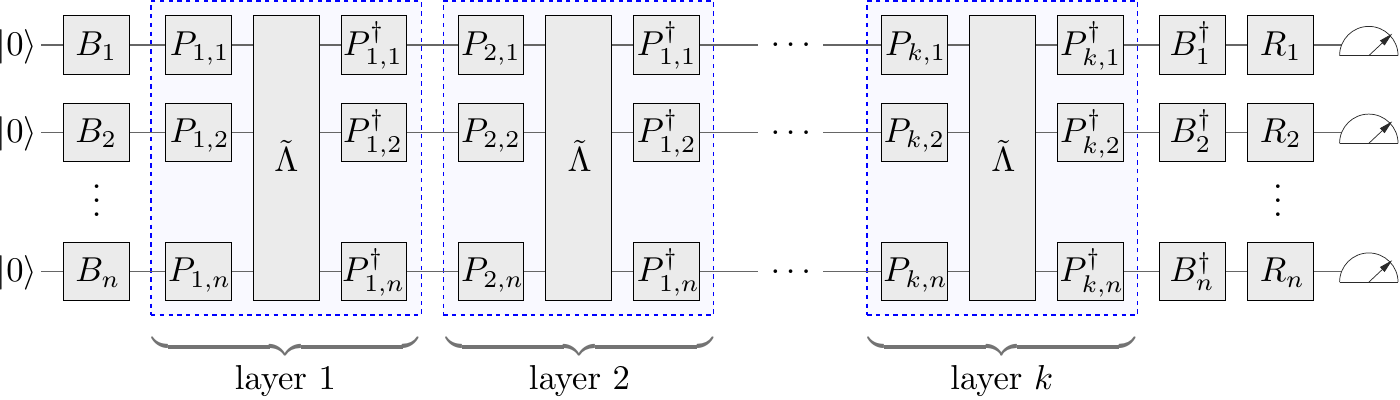}
\caption{Benchmark circuit.}\label{Fig:Benchmark}    
\end{figure}
The estimation of the fidelities that will be needed to reconstruct
the sparse noise model uses a slightly simplified version of the
algorithm proposed in~\cite{FLA2020Wa} and considers the setting where
gates are applied to all qubits. The benchmark circuits are of the
form shown in Figure~\ref{Fig:Benchmark}, where the single-qubit gates
shown are all noiseless. Although this may seems to contradict the
assumption that we only have access to noisy gates, note that the
noise is assumed to be independent of the unitaries applied to each
qubit, which therefore allows us to apply as many consecutive
unitaries as we like with only a single noise term by simply
multiplying the individual unitaries into a single final unitary and
applying the noisy version of this final unitary. That means that
successive gates $B_1$ and $P_{1,1}$ will be combined into some
unitary $U_1$, which is then applied to the circuit along with the
associated noise channel. For convenience we assume that the noise
following the $R_i$ gates appears as readout errors. Having convinced
ourselves that we can actually implement the circuits from
Figure~\ref{Fig:Benchmark}, we now describe the different
components. Gates $B_i$ and $B_i^{\dag}$ implement basis changes
between different Pauli bases. Each cycle consists of the noise
channel $\hat{\Lambda}$ conjugated by random Pauli terms
$P_{i,j}$. When averaged over all possible Pauli terms, this
implements a Pauli twirl of the noise channel $\tilde{\Lambda}$,
resulting in a Pauli channel $\Lambda$, which has a diagonal
Pauli-transfer matrix with fidelity $f_i $ for Pauli $P_i$. The final
gates $R_i$ are sampled uniformly at random from $\{I,X\}$ and are
used in combination with classical post-processing to diagonalize the
readout error~\cite{CHE2020YZFa-arXiv,BER2020MTa-arXiv}. In order to
determine the fidelity $f_i$ we start with the Pauli-Z term $P_{z(i)}$
that has the same support as $P_i$. The initial state $\tilde{\rho}$
is a noisy version of $\ket{0}\bra{0}$ with associated
state-preparation fidelity $s_i = \Tr[P_{z(i)}\tilde{\rho}]$. The
basis change gates $B_i$ change $P_{z(i)}$ to $P_i$, and we then apply
$k$ cycles, each contributing a fidelity term $f_i$. As a result of
diagonalization of the readout errors, we can define a readout
fidelity $r_i = r_{z(i)}$. Overall, this means that the expected value
for the observable $P_i$, measured through observable $P_{z(i)}$ using
bases changes, is given by $(s_i r_i) f_i^k = \alpha_i f_i^k$.
Dividing the estimates obtained for $k$ and zero cycles, then gives an
unbiased estimate of $f^k$, free of state-preparation and readout
errors.

Now that we have access to estimates of individual fidelities of
$\Lambda$, we would like to fit a model that can capture
crosstalk. For this we propose to use a two-local Lindblad model, with
coefficients terms $\mathcal{K}$ given by the union of all unit-weight
Paulis and all weight-two Paulis whose support corresponds to
connected qubits. Given these coefficients we need to determine the
set $\mathcal{B}$ of Paulis for which we estimate the fidelity, such
that the matrix $M(\mathcal{B},\mathcal{K})$ is full rank. For this we
use the result from Section~\ref{Sec:FullRankM}, which shows that
choosing $\mathcal{B} = \mathcal{K}$ results in a square invertible
$M$. With this, the next step is to estimate the fidelities and fit
the model. This is where sampling error comes in: we can only estimate
$\alpha_i f_i^k$ up to an additive error $\epsilon$ that decreases
with the number of circuit instances. In
Section~\ref{Sec:ErrorAnalysis} we therefore study the sample
complexity and the final accuracy of the noise model and its inverse.
For a given circuit it is generally possible to estimate a number of
fidelities. In section~\ref{Sec:Bases} we show, under mild conditions
on the qubit topology, that it suffices to measure in nine different
bases.

\subsection{Fidelities for model fitting}\label{Sec:FullRankM}

We can represent qubit topology as an undirected graph in which each
vertex corresponds to a qubit, and where edges indicate a physical or
logical connection between qubits. For our two-local Lindbladian noise
model we choose model coefficients $\mathcal{K}$ corresponding to
Paulis with support on qubits that are connected by edges, as well as
all Paulis that are supported on subsets of the former supports, which
in this case corresponds to the individual qubits. A direct
consequence of the result below is that
$\mathcal{M}(\mathcal{K},\mathcal{K})$ is full rank. Choosing any set
of benchmark fidelities $\mathcal{B}$ that includes $\mathcal{K}$
gives a full column-rank $M$ and thus ensures that the least-squares
problem has a unique solution.

\begin{theorem}
  Given a set $\{\mathcal{S}_i\}$ of supports
  $\mathcal{S}_i \subseteq [n]$.  Define the set
  $\mathcal{V} = \{\mathcal{V}_j\}$ as the union over $i$ of all
  non-empty subsets of $\mathcal{S}_i$, including the sets
  themselves. For each $j$ let $\mathcal{P}_j$ be the set of all
  $n$-qubit Pauli strings supported on $\mathcal{V}_j$, and let
  $\mathcal{P} = \bigcup_j \mathcal{P}_j$. Then
  $M(\mathcal{P},\mathcal{P})$ is full rank.
\end{theorem}
\begin{proof}
  Since permuting the matrix rows and columns leaves the rank
  unchanged we assume that the sets $\mathcal{V}_j$ are ordered
  according to increasing cardinality and that the Pauli strings in
  each set $\mathcal{P}_j$ are sorted lexicographically.
  Define $V = -2M(\mathcal{P},\mathcal{P})$ and partition the matrix
  into blocks such that block
  $V(i,j) = -2M(\mathcal{P}_i,\mathcal{P}_j)$. These blocks are
  concisely expressed as $V(i,j) = O(i,j,S) - O(i,j,\mathbb{1})$,
  where $\mathbb{1}$ is a $3$-by-$3$ matrix of ones, and
\begin{equation}\label{Eq:OS}
O(i,j,op) = \bigotimes_{\ell=1}^n \begin{cases}
op & \ell\in \mathcal{V}_i \cap \mathcal{V}_j\\
\mathbf{e}  & \ell \in \mathcal{V}_i \setminus \mathcal{V}_j\\
\mathbf{e}^T & \ell \in \mathcal{V}_j \setminus \mathcal{V}_i\\
1 & \ell\not\in \mathcal{V}_j  \cup \mathcal{V}_j
\end{cases},
\qquad
S = \left(
\begin{array}{rrr} 1 & -1 & -1\\ -1 & 1 & -1 \\ -1 & -1 & 1\end{array}
\right),
\end{equation}
and $\bf{e}$ denotes a column vector of ones of length three with
transpose $\bf{e}^T$. Note that matrices $S$ and $Q:= S - \mathbb{1}$
are invertible. We prove invertibility of $V$ by reducing it to a
block-diagonal matrix with invertible block by iteratively applying
sweep operations. Sweeping of the blocks in row $i$ or column $j$ by
those $k$ is done only when $\mathcal{V}_k \subset\mathcal{V}_i$ or
$\mathcal{V}_k \subset\mathcal{V}_j$, respectively. The sweep
operations are defined as
\begin{align}
\begin{split}
\mbox{row\_sweep$(i,k$):}\ \ & V(i,j) \gets V(i,j) - O(i,k,I)V(k,j),\\
\mbox{column\_sweep($j,k)$:}\ \  & V(i,j) \gets V(i,j) - V(i,k)O(k,j,I).
\end{split}\label{Eq:SweepOp}
\end{align}
The structure of $O(i,j,op)$ in terms of the locations of matrices and
additional terms $\bf{e}$, $\bf{e}^T$, and $1$ is prescribed by the
sets $\mathcal{V}_i$ and $\mathcal{V}_j$, and we now show that
$V(i,j)$ can always be written as a sum of tensors sharing the same
structure. This is immediate for the initial $V(i,j)$ and we therefore
focus on the updates in~\eqref{Eq:SweepOp}. For the row update we
consider the term $O(i,k,I)V(k,j)$, or, since we are only interested
in structure, $O(i,k,I)O(k,j,op)$. By writing out a table of terms
based on membership of $\ell$ in $\mathcal{V}_i$, $\mathcal{V}_j$, and
$\mathcal{V}_k$, with the constraint that
$\mathcal{V}_k \subset \mathcal{V}_i$, it can easily be verified that
this indeed holds. Of special interest is the case where
$\ell\in\mathcal{V}_{i,j}$ and $\ell\not\in\mathcal{V}_{k}$. In this
case the $\ell$th terms in $O(i,k,I)$ and $O(k,j,op)$ are $\bf{e}$ and
$\bf{e}^T$ respectively, which means that their product is the matrix
$\mathbb{1}$. The same approach shows that column sweeps also maintain
the structure. For convenience we represent by $B^{(\ell)}(i,j)$ the
value of block $V(i,j)$ at iteration $\ell$ as a sum of matrix
products of $w_{i,j}$ matrices, thus omitting the fixed non-matrix
terms in the full representation. As seen above, the initial block
values are given by
$B^{(0)}(i,j) = S^{\otimes w_{i,j}} - \mathbb{1}^{\otimes w_i}$, where
we define $A^{\otimes 0} = 1$.  We provide a sweeping algorithm such
that
\begin{equation}\label{Eq:Bij}
B^{(\ell)}(i,j) = \begin{cases}
Q^{\otimes w_{i,j}} & \mbox{if $i=j$ and $w_{i,j} \leq\ell$},\\
0 & \mbox{if $i\neq j$ and $w_{i,j} \leq \ell$},\\
Q^{\otimes w_{i,j}} & \mbox{if $w_{i,j} = \ell+1$},\\
S^{\otimes w_{i,j}} - \mathcal{Q}(w_{i,j},\ell) & \mbox{if $w_{i,j} > \ell+1$},
\end{cases}
\end{equation}
holds for all $\ell \geq 0$, where $\mathcal{Q}(w,\ell)$ is the sum of
matrices $\{\mathbb{1},Q\}^{\otimes w}$ that contain at most $\ell$ terms
equal to $Q$.
The expressions for the second and third case of~\eqref{Eq:Bij} are
special cases of $S^{\otimes
  w_{i,j}}-\mathcal{Q}(w_{i,j},\ell)$.
Namely, we have
$S^{\otimes w} = (\mathbb{1} + Q)^{\otimes w} = \mathcal{Q}(w,w)$,
while for $w \geq 1$ it holds that
$\mathcal{Q}(w,w) - \mathcal{Q}(w,w-1) = Q^{\otimes w}$, and therefore
\[
S^{\otimes w} - \mathcal{Q}(w,w-1) = Q^{\otimes w}
\quad\mbox{and}\quad
S^{\otimes w} - \mathcal{Q}(w,\ell \geq w) = 0.
\]
It can be verified that~\eqref{Eq:Bij} holds for $\ell=0$ by observing
that $\mathcal{Q}(w,0) = \mathbb{1}^{\otimes w}$. Assume
that~\eqref{Eq:Bij} holds for some $\ell \geq 0$. Then the blocks can
be arranged as shown in Figure~\ref{Fig:BlocksB}. Because
$w_{i,j} = w_{j,i} \leq w_i$ we can ignore all rows $i$ with
$w_i \leq \ell$, and likewise for the columns, since
$B^{(\ell)}(i,j) = B^{(\ell)}(j,i)^T$.  Submatrix A consists of all
blocks $(i,j)$ with $w_i\! =\! w_j\! =\!\ell + 1$, and it follows from
$w_{i,i} = w_i$ for the diagonal blocks and $w_{i,j} < w_i$ for the
off-diagonal blocks that A is block diagonal with
$B^{(\ell)}(i,i) = Q^{\otimes \ell+1}$. Submatrix B satisfies
$w_j = \ell+1$ and $w_i > \ell+1$ and therefore also has blocks with
weight at most $\ell+1$. This means that the blocks are again either
$Q^{\otimes \ell+1}$ or zero, and likewise for submatrix C. Given this
structure it is easily seen that we can clear block $B$ by sweeping
rows $i \in \mathcal{I} = \{i \mid w_i > \ell+1\}$ with rows
$k \in \mathcal{K} = \{k \mid w_k = \ell+1\}$.  Consider an arbitrary
row $i \in\mathcal{I}$. For each $k \in\mathcal{K}$ we have
$\vert\mathcal{V}_k\vert = \ell+1$, and it therefore follows from
$\vert\mathcal{V}_i\vert = w_i$ and the assumption that all non-empty
subsets of $\mathcal{V}_i$ are present, that there are exactly
$ \binom{w_i}{\ell+1}$ elements $k\in \mathcal{K}$ for which
$\mathcal{V}_k \subset \mathcal{V}_i$ and $w_{i,k} = \ell+1$. We
therefore need to sweep row $i$ with precisely these $k$ values, which
we denote by $\mathcal{K}'$ with implicit dependency on $i$ and
$\ell$. For the effect on block D, consider a block $(i,j)$ with an
arbitrary $j$ for which $w_j > \ell+1$. If $w_{i,j} \leq \ell$ we have
$B^{(\ell)}(i,j) = 0$ and it follows from
$\mathcal{V}_k \subset \mathcal{V}_i$ that $w_{k,j} \leq \ell$ for all
$k\in\mathcal{K}'$, which means that all updates to it are zero as
well. For block $B^{(\ell)}(i,k)$ in C to be nonzero we must have
$\mathcal{V}_k \subset \mathcal{V}_j$, which is the case for exactly
$ \binom{w_{i,j}}{\ell+1}$ values of $k\in\mathcal{K}'$ since
$\vert\mathcal{V}_i \cap \mathcal{V}_j\vert = w_{i,j}$. All elements
$k$ outside this set, say $\mathcal{K}''$, will have
$\vert \mathcal{V}_j \cap \mathcal{V}_k\vert < w_i = \ell+1$ and
therefore correspond to a zero block $(k,j)$. Because each block
$(k,j)$ for $k\in\mathcal{K}''$ is equal to $Q^{\otimes\ell+1}$ and
multiplied by $O(i,j,I)$, we conclude that block $(i,j)$ is swept by
all matrices $\{\mathbb{1},Q\}^{\otimes w_{i,j}}$ with exactly
$\ell+1$ terms equal to $Q$. That means that $B^{(\ell+1)}(i,j)$ is
updated to $S^{\otimes w_{i,j}} - \mathcal{Q}(w_{i,j},\ell+1)$. Once
submatrix B is cleared we can repeat the same set of sweeps over the
column indices. Since all blocks in B are zero this does not affect
submatrix D and only zeros out all blocks in C. The claim
that~\eqref{Eq:Bij} holds for all $\ell \geq 0$ then follows directly
by induction. For $\ell = \max_i(w_i)$ we see that the diagonal
elements $B^{(\ell)}(i,i)$ are $Q^{\otimes w_i}$ and invertible, and
all off-diagonal elements are zero, as required.

\begin{figure}[!t]
\centering
\includegraphics[width=0.3\textwidth]{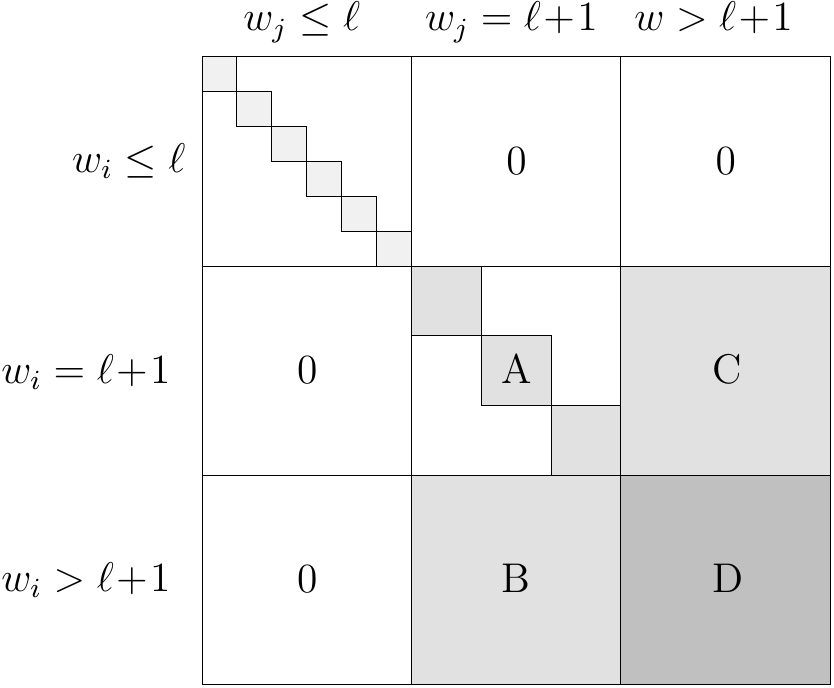}
\caption{Structure of the blocks $B^{(\ell)}(i,j)$ at iteration $\ell$.}\label{Fig:BlocksB}
\end{figure}
\end{proof}

\subsection{Sample complexity and error analysis}\label{Sec:ErrorAnalysis}
In order to fit the Lindbladian noise model we need to estimate
individual fidelities $f_i$. We refer to the set of Paulis that are
measured in the experiment as $\mathcal{B}$. The task at hand is to
reconstruct the full Pauli-Lindblad model $\Lambda$ by only measuring
a sparse subset $|\mathcal{B}| \ll 4^n-1$ of fidelities and then
fitting the model to determine the parameters $\lambda_k$ with
$k \in \mathcal{K}$. The deviation of these parameters from the
assumed ground truth is bounded in~\eqref{Eq:paramBND}. The error
bound in the following theorem \ref{Thm:Channel} bounds the deviation
of all fidelities of the ground truth Pauli-Lindblad model $\Lambda$
from the model fidelities obtained parameter estimates for
$\lambda_k$. The estimation of the directly measured fidelities from
$\mathcal{B}$ is done using random circuit instances of the form shown
in Figure~\ref{Fig:Benchmark} for various cycle lengths $k$. For a
fixed $k$, the expected value for observable $P_i \in \mathcal{B}$,
measured using appropriate basis changes and readout twirling, is
given by $\alpha_i f_i^k$. Measuring a single shot for each qubit for
a single circuit instance is equivalent to sampling an element from a
distribution over $\{-1,1\}$ with expectation value $\alpha_i
f_i^k$. For the deviation from the expected value we can apply
Hoeffding's inequality, which states that for given $N$ independent
random variables $X_j$ sampled from any distribution $[-\beta,\beta]$,
the deviation of $\bar{X} = N^{-1}\sum_{i=1}^N X_i$ to the expected
value satisfies
\begin{equation}\label{Eq:Hoeffding}
\Pr\left(\big\vert \bar{X} - \mathbb{E}(X)\big\vert \geq
  \epsilon\right) \leq 2\exp\left(-\frac{N\epsilon^2}{2\beta}\right).
\end{equation}
From this it follows that, by taking
$N \geq 2\log(2/\delta)/\epsilon^2$ samples, the estimate
$\alpha_if_i^k + \epsilon_{i,k}$ satisfies
$\vert\epsilon_{i,k}\vert \leq \epsilon$ with probability at least
$1-\delta$. The number of samples in this case corresponds to the
number of circuit instances. We will revisit the sample complexity
below, but first state the following result assuming sufficiently
accurate samples:

\begin{theorem}\label{Thm:Channel}
  Denote the Pauli terms in a given Pauli-Lindblad channel by
  $\mathcal{K}$, and assume we have benchmark fidelities $\mathcal{B}$
  such that $f_i \geq 1/2$ for all $i\in\mathcal{B}$ and
  $M = \mathcal{M}(\mathcal{B},\mathcal{K})$ is full column rank. Let
  $k\geq 1$ be an integer such that $f_i^k \geq 1/2$ for all
  $i\in\mathcal{B}$, and assume that the readout and sampling errors
  satisfy $\alpha_i \geq 1/2$ and
  $\vert \epsilon_{i,\ell}\vert \leq \epsilon < 1/4$ for all
  $i\in\mathcal{B}$ and $\ell\in\{0,k\}$. Then the estimated inverse
  channel fidelities $\widehat{(f_j^{-1})}$ for any $j$ and scaling
  factor $\hat{\gamma}$ satisfy
\begin{equation}\label{Eq:ChannelBounds}
C_{\epsilon}^{-\tau}
\leq f_j(\widehat{f_j^{-1}})\leq 
C_{\epsilon}^{\tau}
\qquad\mbox{and}\qquad
\gamma C_{\epsilon}^{-\tau}
\leq \hat{\gamma}\leq
\gamma C_{\epsilon}^{\tau},
\end{equation}
where $C_{\epsilon} = (1+4\epsilon) / (1-4\epsilon)$ and
$\tau = \sqrt{\vert\mathcal{K}\vert\cdot\vert\mathcal{B}\vert}/ (\sigma_{min}(M)k)$.
\end{theorem}
\begin{proof}
  The analysis follows the error bounds on the measured fidelities by
  Flammia and Wallman in~\cite{FLA2020Wa}. The protocol estimates the
  fidelity based on sampled values for $\alpha_if_i^{\ell}$ for a pair
  of depths $\ell\in \{0,k\}$. Given that the additive error in the
  sampled values is bounded by $\epsilon$, the estimated fidelity
  $\hat{f}_i$ satisfies
\[
\frac{\alpha_if_i^k - \epsilon}{\alpha_if_i^0 + \epsilon}
\leq \hat{f}_i^k\leq
\frac{\alpha_if_i^k + \epsilon}{\alpha_if_i^0 - \epsilon}.
\]
Dividing the enumerator and denominator by $\alpha_i$ and using the
assumption that $\alpha_i \geq 1/2$ and $f_i^k \geq 1/2$, gives
\[
f_i^k\frac{1 - 4\epsilon}{1+2\epsilon} \leq \hat{f}_i^{k}
\leq f_i^k\frac{1 + 4\epsilon}{1-2\epsilon}.
\]
By relaxing the denominator, taking the logarithm, and reorganizing we
obtain
\[
\log C_{\epsilon}^{-1} \leq k\log(\hat{f}_i) - k\log(f_i) \leq \log
C_{\epsilon},
\]
and therefore
\begin{equation}\label{Eq:EstLogFidelity}
\big\vert \log(\hat{f}_i) - \log(f_i)\big\vert\leq
\log(C_{\epsilon})/k.
\end{equation}
In order to solve the least-squares problem in Eq.~\eqref{Eq:NNLS} we
need to estimate the fidelities in the set $\mathcal{B}$. Given the
bound on the elementwise error in Eq.~\eqref{Eq:EstLogFidelity}, we
can bound the two-norm of the vector of log fidelities of length by
$\sqrt{\vert\mathcal{B}\vert}\log(C_{\epsilon})/k$. To bound the error
in the estimated parameters $\hat{\lambda}$ we use
Theorem~\ref{Lemma:LSQR}, below, which gives
\begin{equation}\label{Eq:paramBND}
\Vert \hat{\lambda} - \lambda\Vert_2
\leq
\frac{\log(C_{\epsilon})\sqrt{\vert\mathcal{B}\vert}}{2\sigma_{\min}(M)k}\;.
\end{equation}

For the scaling parameter $\gamma$, we use
$\Vert \hat{\lambda}-\lambda\Vert_1 \leq
\sqrt{\vert\mathcal{K}\vert}\Vert\hat{\lambda}-\lambda\Vert_2$ and
Eq.~\eqref{Eq:Eta} to get the right-hand side of
Eq.~\eqref{Eq:ChannelBounds}.

We invert the estimated channel by flipping the sign, and obtain the
log fidelities by multiplication with $M$. Given that the entries in
$M$ are either zero or one, and that the estimated coefficients
$\hat{\lambda}$ are all nonnegative, the deviation in
$M(\lambda-\hat{\lambda})$ is bounded again by
$\Vert\lambda-\hat{\lambda}\Vert_1$. Multiplying by the factor of two
that appears in Eq.~\eqref{Eq:Mf}, we thus have
\[
\vert\log(\widehat{f_i^{-1}}) - \log(f_i^{-1}) \vert \leq 
\tau\log(C_{\epsilon}),
\]
which is easily rewritten to obtain the left-hand side
of Eq.~\eqref{Eq:ChannelBounds}.
\end{proof}

\begin{lemma}\label{Lemma:LSQR}
  Given a closed convex set $\mathcal{C}$, and full column rank matrix
  $A \in \mathbb{C}^{s\times t}$ with singular value decomposition
  $A=U\Sigma V^*$. Then the solution to the constrained least-squares
  problem:
\[
x(b) = \mathop{\mathrm{argmin}}_{x\in\mathcal{C}}\quad \half\norm{Ax-b}_2^2.
\]
satisfies
\begin{equation}\label{Eq:LSQRBound1}
\norm{x(b_1) - x(b_2)}_2 \leq \frac{1}{\sigma_{\min}(A)}\norm{U^*(b_1 - b_2)}_2.
\end{equation}
\end{lemma}
\begin{proof}
Define $\mathcal{C}' = A\mathcal{C} = \{Ac \mid c\in\mathcal{C}\}$ and
\[
y(b) = \mathop{\mathrm{argmin}}_{y\in\mathcal{C}'}\ \half\norm{y - b}_2^2 =: \mathcal{P}_{\mathcal{C}'}(b)
\]
There is a one-to-one correspondence between points in $\mathcal{C}'$
and $\mathcal{C}$, and for the solution we have $y(b) = Ax(b)$.
Moreover, because $\mathcal{C}'$ lies in the subspace spanned by
$U \in \mathbb{C}^{s\times t}$, we have
$\mathcal{P}_{\mathcal{C}'}(b) = \mathcal{P}_{\mathcal{C}'}(UU^*b)$.
It then follows from the fact Euclidean projection onto a convex set
($\mathcal{P}_{\mathcal{C}'}$) is non-expansive, that
\begin{align*}
\sigma_{\min}(A)\norm{x(b_1) - x(b_2)}_2
&\leq \norm{y(b_1) - y(b_2)}_2\\
& = \norm{\mathcal{P}_{\mathcal{C}'}(b_1) - \mathcal{P}_{\mathcal{C}'}(b_2)}_2\\
&= \norm{\mathcal{P}_{\mathcal{C}'}(UU^*b_1) - \mathcal{P}_{\mathcal{C}'}(UU^*b_2)}_2\\
& \leq \norm{UU^*(b_1 - b_2)}_2\\
& = \norm{U^*(b_1 - b_2)}_2. 
\end{align*}
\end{proof}

\subsubsection{Measurement bases}\label{Sec:Bases}

In our discussion so far we considered the estimation of individual
fidelities by sampling random circuit instances and processing their
measurements. However, given a single basis it is possible to estimate
a large number of fidelities using the same measurements. When
considering a two-local Pauli-Lindblad noise model it suffices to
consider all of the nine $\{X,Y,Z\}^{\otimes 2}$ bases on each qubit
pair. Under some mild conditions on the qubit topology, we now show
that is suffices to measure using a total of nine bases. That is,
there exist nine Pauli strings such that the substrings corresponding
to a pair of connected qubits cover all nine local bases.

\begin{theorem}\label{Thm:Combine}
  Given a qubit topology $(\mathcal{V},\mathcal{E})$ whose vertices
  are ordered in such a way that no vertex $v\in\mathcal{V}$ is
  preceded by more than two connected vertices. Then there exist nine
  Pauli strings such that for each $(v_i,v_j) \in \mathcal{E}$ the
  substrings at locations $v_i$ and $v_j$ exactly cover
  $\{X,Y,Z\}^{\otimes 2}$.
\end{theorem}
\begin{proof}
  Given a vertex $v_i$, there are three cases to consider. In the
  first case, none of the predecessors of $v_i$ is connected to $v_i$
  and we simply assign a random permutation of three instances of $X$,
  $Y$, and $Z$ to location $v_i$ of the strings. In the second case,
  $v_i$ is connected to exactly one predecessor, vertex $v_j$. We
  assign a random permutation of $X$, $Y$, and $Z$ to the string
  location $v_i$ for those strings where $v_j$ is equal to $X$, and
  repeat the same for $Y$ and $Z$. In the third case $v_i$ is
  connected to two predecessors, $v_j$ and $v_k$. Assuming, without
  loss of generality that the strings are ordered such that the first
  three strings have $X$ at location $v_j$, followed by three strings
  with $Y$ and then three strings with $Z$. We can freely reorder the
  groups of three strings as well as the strings within each
  group. The possible values for $v_k$ can then always be reordered to
  those given in Figure~\ref{Fig:Complete}, where the string values at
  location $v_k$ are indicated by shades of gray. The figure also
  provides an example assignment for Pauli character at the current
  location, $v_i$, such that that each block of three as we all each
  shade of gray contains each of $X$, $Y$, or $Z$ exactly once. It
  follows that the substrings of locations $(v_i,v_j)$ and $(v_i,v_k)$
  contain the required nine strings of length two, as desired.
\end{proof}

\begin{figure}[t]
\centering
\includegraphics[width=0.95\textwidth]{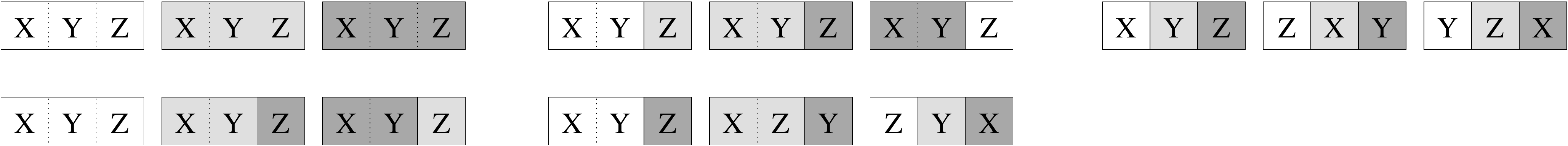}
\caption{Assigning Pauli bases for a vertex connected to two vertices
  for which the bases are already fixed. Each group of three locations
  corresponds to a permutation of $X$, $Y$, and $Z$ values for the
  first node. The shades of gray represent $X$, $Y$, and $Z$ in some
  order for the second node. Example assignments for the basis values
  of the current vertex are as shown. Note that each block of three
  locations, as well as each shade of gray, contains each basis
  exactly once.}\label{Fig:Complete}
\end{figure}

Another way to view the conditions for Theorem~\ref{Thm:Combine} is
that we iteratively visit vertices such that no more than two
connected vertices has already been visited. This condition applies
for commonly used two-dimensional grid and heavy-hexagon topologies.
For a regular two-dimensional grid this can be done in a left-to-right
and top-to-bottom fashion. The vertices of the heavy-hexagon topology
have a maximum degree of three and no two such vertices are
connected. As a very simple algorithm we could, for instance, first
sample values for the isolated vertices with degree three, which then
leaves only vertices of degree one or two, which are then easily
completed.

\subsubsection{Overall noise-learning complexity}

When using Hoeffding's inequality~\eqref{Eq:Hoeffding} we can select
the probability $\delta$ with which the estimated fidelity
exceeds~$\epsilon$. When considering $K$ different fidelity estimates,
each with failure probability $\delta$ and possibly correlated, it
follows from the union bound that the probability that at least one
fails is bounded above by $K\delta$. This means that all
$K = \vert\mathcal{B}\vert$ fidelity estimates are simultaneously
$\epsilon$ accurate with probability at least $1 - K\delta$,
regardless of whether they are estimated independently or using the
shared sampled obtained for the nine different bases as described
above. For a desired overall success probability of $1 - \delta'$ it
thus suffices to choose $\delta = \delta'/\vert\mathcal{B}\vert$.
Substitution in Eq.~\eqref{Eq:Hoeffding} and rearranging then gives a
sample complexity of
\[
N \geq \frac{2\log(2\vert\mathcal{B}\vert/\delta')}{\epsilon^2}
\]
circuit instances per basis. In case we use nine bases, each with
depths zero and $k$, this gives a total number of $18\lceil N\rceil$
circuit instances. The value of $k$ may not be known in advance, but
we may select a value $k_{\max}$ and then use a binary search to find
the largest $k$ for which all fidelities are above $1/2 + \epsilon$.
This takes at most $\lceil\log_2(k_{\max})\rceil$ trials. For these to
all succeed with probability at least $1-\delta'$, we can choose
$\delta = \delta' / (\vert\mathcal{B}\vert
\cdot\lceil\log_2(k_{\max})\rceil)$.

\section{Noise learning for two-qubit Clifford gates with crosstalk} \label{Sec:TwoQubitGates}

The results from the previous section also apply to noise channels
associated with layers of arbitrary Clifford gates. For instance, we
may have a layer of controlled-{\sc{not}} (CX) or
controlled-phase~(CZ) operations whose implementation is subject to
noise. Twirling the associated noise is possible by adding pairs of
Pauli operations before and after the operation such that the second
Pauli equals the first up to conjugation by the ideal Clifford
operator associated with the layer, up to a global phase. Learning
procedures of noise in such circuit families for more general Pauli
channels have been derived
in~\cite{ERH2019WPMa,PhysRevX.4.011050,HEL2019XVWa}. Given estimates
of all fidelities in $\mathcal{B}$ we can fit the noise model and
apply error mitigation with the same theoretical guarantees without
any change.  The one significant difference from the single-qubit
scenario, however, lies in the benchmarking process to estimate the
fidelities.

Assuming the noise channel $\Lambda$ of a noisy CZ gate has been
twirled to a Pauli channel, we can then consider the fidelity of Pauli
IX. This Pauli is one of the different components of the initial state
after applying a ZX basis change, obtained by applying a Hadamard gate
on the second qubit. Given that $\Lambda$ is diagonal in the Pauli
basis, applying the noise channel incurs a multiplicative fidelity
term $f_{IX}$, while leaving the Pauli term itself unchanged. Applying
the ideal CZ gate corresponds to conjugation with the CZ operator,
which changes IX to ZX. For the second application of the noisy CZ
gate we first apply the noise channel $\Lambda$, which now incurs an
$f_{ZX}$ fidelity term since the current Pauli is ZX. Finally,
applying the second ideal CZ gate changes the Pauli back to the
initial IX.  Repeated application of the gate, as before, may
therefore give rise to exponentiated products of terms, such as
$(f_{IX}f_{ZX})^k$. This process, along with the Pauli-transfer
diagram for CZ, can be illustrated as follows:
\begin{center}
\begin{tabular}{ccc}
\raisebox{21pt}{\includegraphics[width=0.40\textwidth]{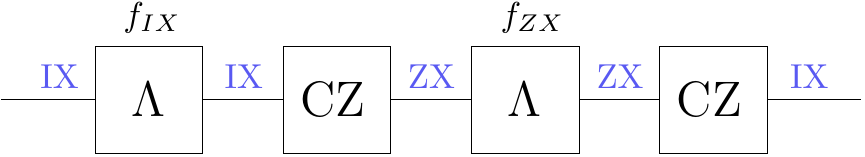}}& \hspace*{20pt}&
\includegraphics[width=0.20\textwidth]{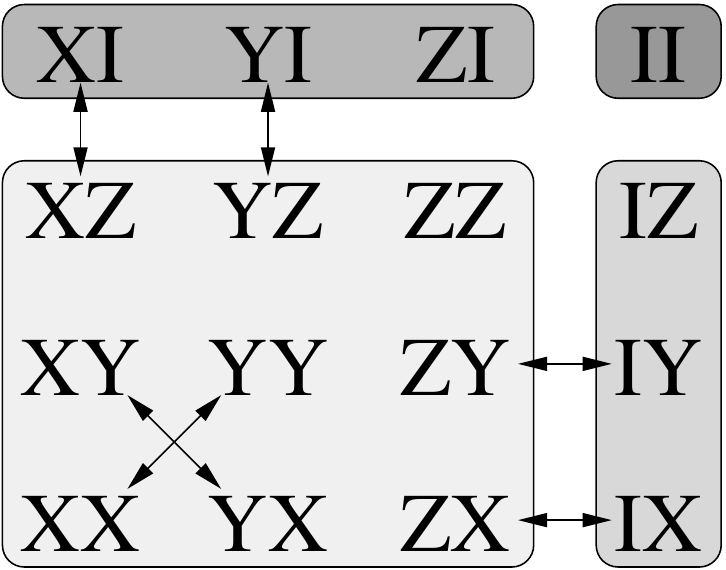}
\end{tabular}
\end{center}
For Pauli terms that are invariant under conjugation by CZ, such at
IZ, and ZI, we obtain powers of the individual fidelities
themselves. For other Paulis that are not invariant, such as XX, we
can engineer powers of the associated fidelities by inserting
additional single-qubit gates after the noisy gate of interest (see
Figure~\ref{Fig:Layers2}e). For instance, for XX we can map the
resulting Pauli YY back to XX by applying phase gates. Note that this
is possible only if application of the gate does not change the
support of the Pauli. For the Pauli pairs indicated by the horizontal
and vertical arrows in the transfer diagram, including the IX-ZX pair
discussed earlier, we cannot resolve individual fidelities this
way. However, given only products of fidelities complicates extracting
individual fidelities: the equality
$\sfrac{1}{\alpha}f_{IX}\cdot\alpha f_{ZX} = f_{IX}f_{ZX}$ holds for
all nonzero values of $\alpha$.

There are various ways of dealing with this degeneracy. The first
approach is to assume that the two fidelities appearing as a pair are
equal. This assumption, which we refer to as the symmetry assumption
throughout this work, allows us to use existing benchmark results and
directly extract the desired fidelities from the estimated cross terms
by simply taking the square root of the product. To motivate this,
consider the Lindblad evolution using a Hamiltonian $H=\pi\cdot CZ/2$,
where $CZ$ denotes the $4\times 4$ matrix representation of the CZ
operator in the standard basis. When setting the diffusive part of the
Lindbladian to a Pauli channel, we observed in preliminary simulations
that conjugate Pauli pairs under the time evolution of the Lindbladian
(which implements the noisy CZ operation) have the same fidelity. This
also applies for resolvable fidelities, as seen in
Figure~\ref{fig:dd-learning-1} for CX gates.

In randomized benchmarking it is common to assume that certain gates,
such as Clifford gates are subject to the same noise channel. As a
second approach, we could therefore make the reasonable assumption
that CZ and CX gates are affected by the same noise. Given that the CZ
gate is implemented as CX conjugated by $(I\otimes H)$, we have
\begin{align*}
C_X\Lambda
& = \tilde{C}_X\\
& = (I\otimes H)\tilde{C}_Z(I\otimes H)\\
& = (I\otimes H)C_Z\Lambda(I\otimes H)\\
& = \underbrace{(I\otimes H)C_Z(I\otimes H)}_{C_X}(I\otimes H)\Lambda(I\otimes H).
\end{align*}
We must therefore have that
$\Lambda = (I\otimes H)\Lambda(I\otimes H)$.  This implies that the
fidelities for $P_1X_2$ and $P_1Z_2$ are the same for any Pauli $P$ on
the first qubit. For the CZ gate this would amount to the assumption
that $f_{IX} = f_{IZ}$, and likewise for the remaining three pairs of
cross terms. Given that we can learn $f_{IZ}$, $f_{XX}$, $f_{YX}$ and
$f_{ZZ}$, we can use this assumption to then infer the fidelities
$f_{IX}$, $f_{XZ}$, $f_{YZ}$, and $f_ZX$.

A third option is to estimate individual fidelities by applying the
noisy gate only once. The main difficulty here is that the initial and
final Pauli component are generally no longer the same. Consequently,
the readout-error correction achieved by dividing with the appropriate
zero-depth fidelity~\cite{BER2020MTa-arXiv} can only remove the SPAM
errors completely when the initial state is exactly the ground state
$\ket{0}$. We consider the topic of finding alternative techniques
that can accurately estimate the individual fidelities for two-qubit
gates as an important topic for future work.

Given that most of our fidelity estimates now come in pairs we no
longer have access to a vector of individual fidelities $f$, but
rather have the elementwise product of vectors $f_1$ and $f_2$. Given
the Pauli terms corresponding to the entries in the vectors we can
form binary matrices $M_1$ and $M_2$. In the ideal case we then have
that $M_1\lambda = -\log(f_1)/2$ and $M_2\lambda =
-\log(f_2)/2$. Adding the two it follows that for pairwise products we
have $(M_1+M_2)\lambda = -\log(f_1\cdot f_2)/2$, where $\cdot$ denotes
elementwise multiplication. We can again obtain the model parameters
$\lambda$ by solving a nonnegative least-squares problem, this time
with $M = M_1+M_2$ and $f = f_1\cdot f_2$. The model parameters are
again unique when $M_1+M_2$ is full column rank, and we next consider
conditions on the measured fidelity pairs that guarantee this.

\begin{figure*}[t]
\centering\setlength{\tabcolsep}{8pt}
\begin{tabular}{ccccc}
\includegraphics[height=74pt]{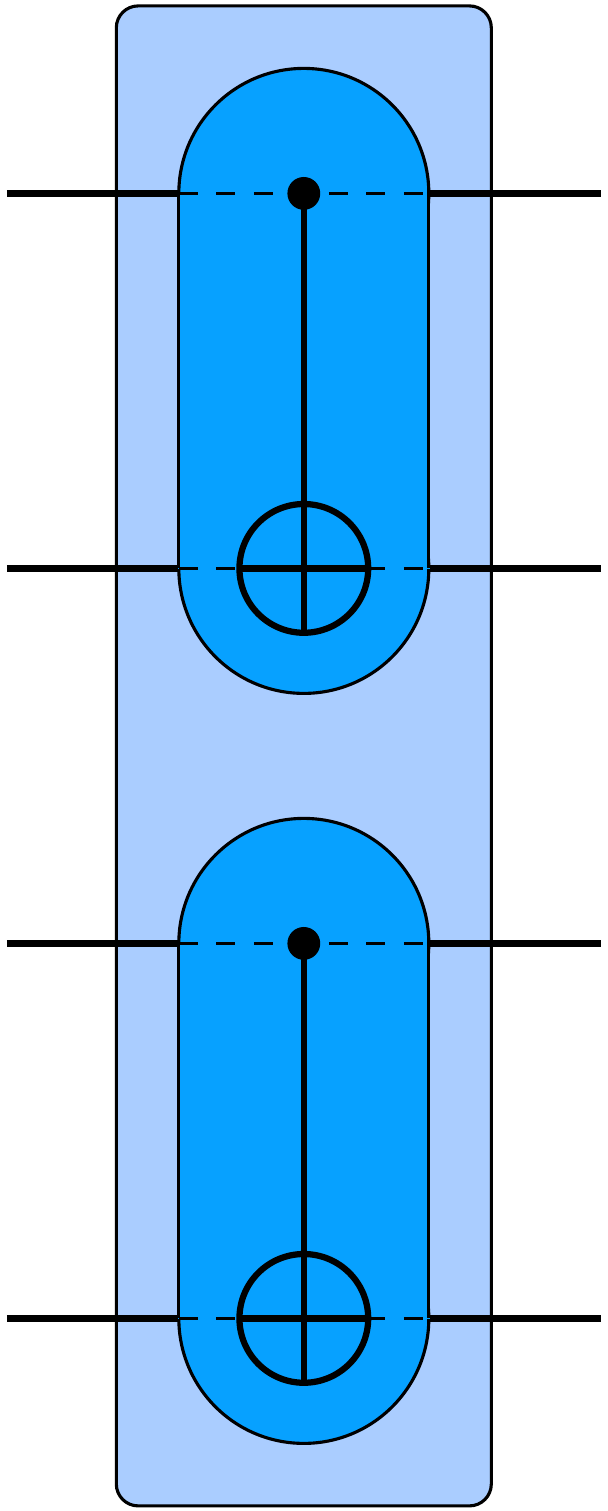}
&
\includegraphics[height=74pt]{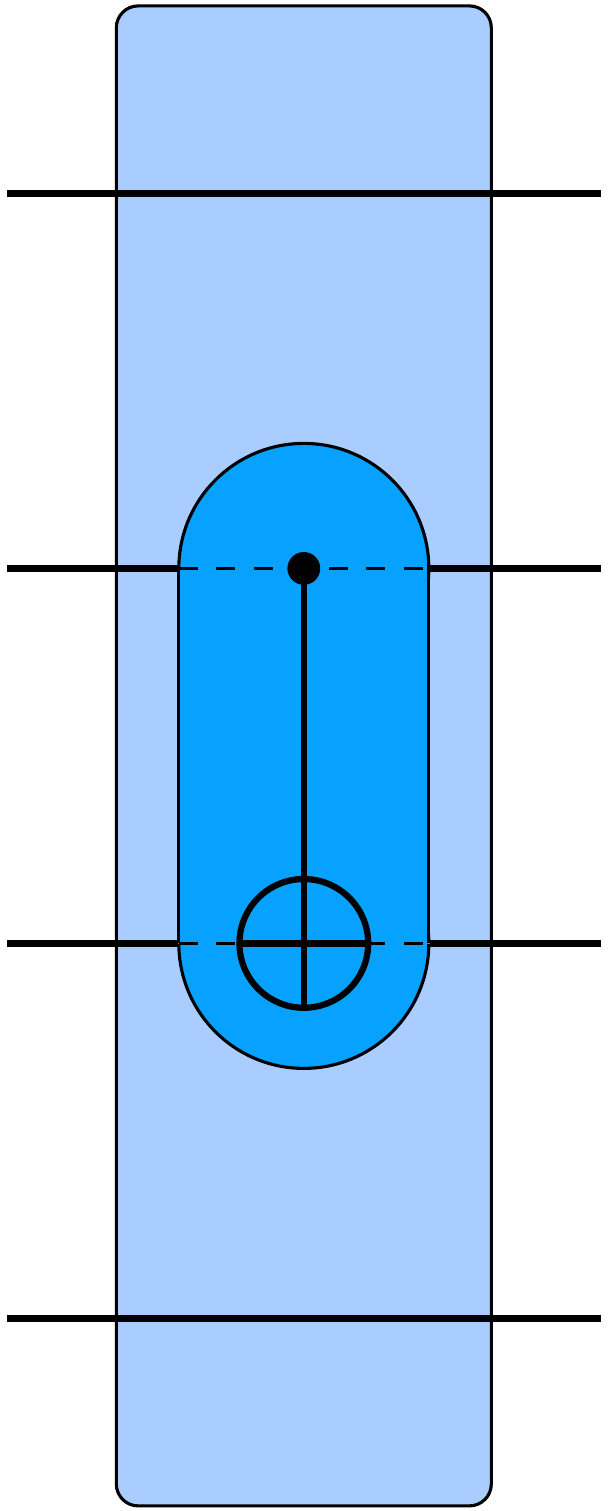}
&
\includegraphics[height=74pt]{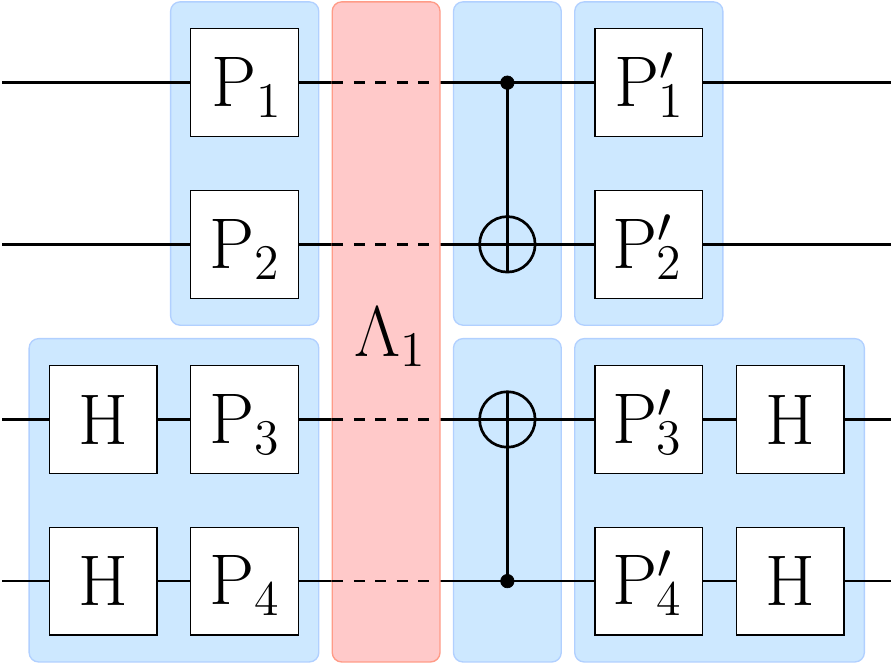}
&
\includegraphics[height=74pt]{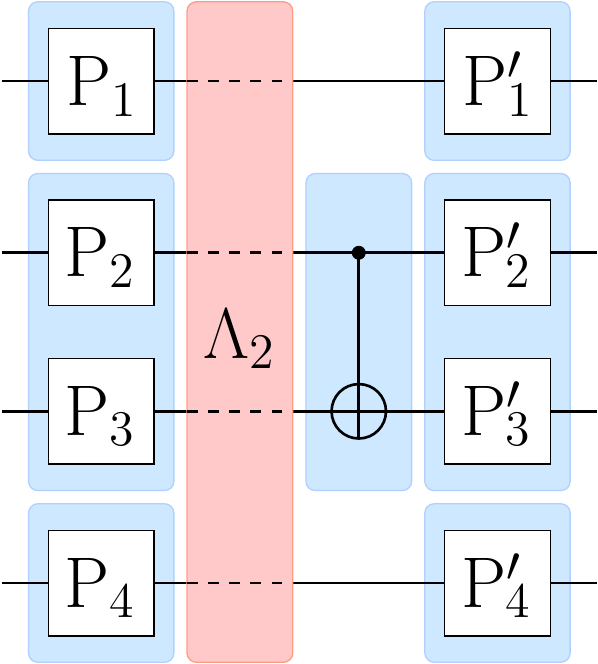}
&
\raisebox{-10pt}{\includegraphics[height=86pt]{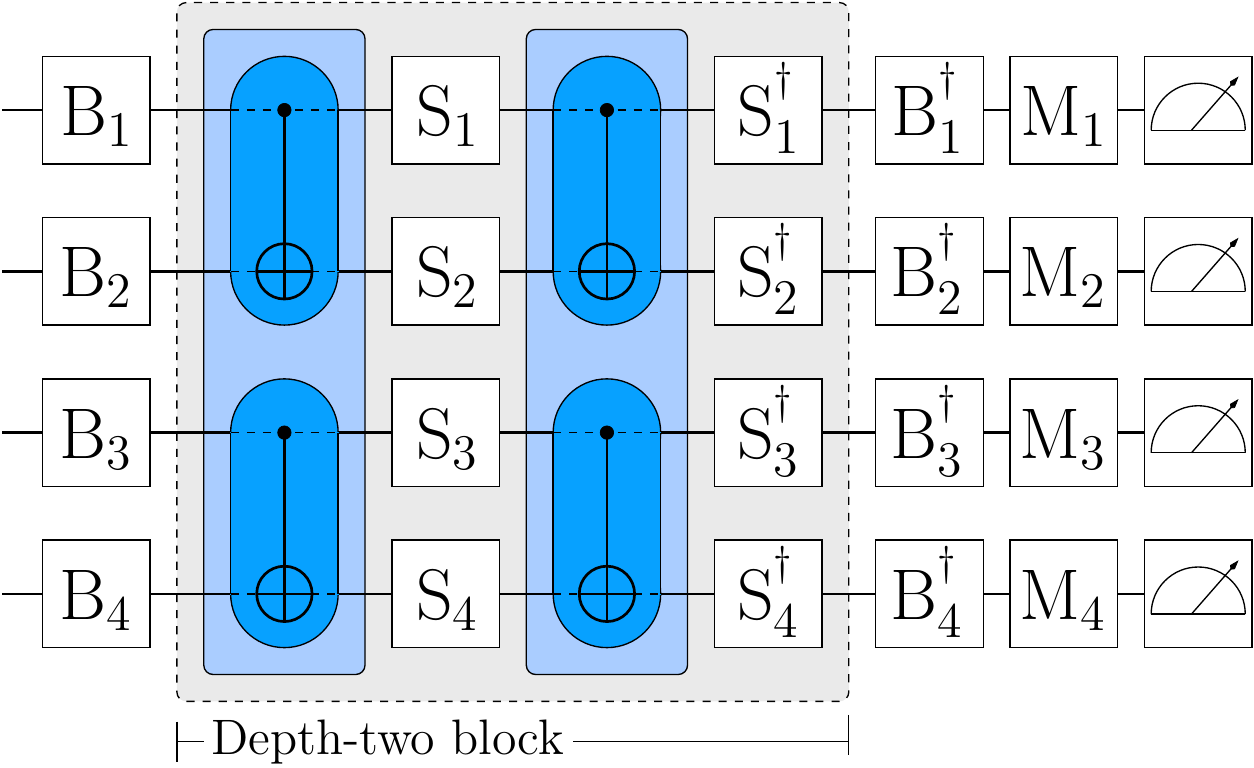}}\\[6pt]
{\bf{a}} & {\bf{b}} & {\bf{c}} & {\bf{d}} & {\bf{e}}\\
\end{tabular}
\caption{\textbf{Overview of circuit layers and noise-learning
    circuits.} Example four-qubit layers with ({\bf{a}}) two {\sc{cx}}
  gates and ({\bf{b}}) a single {\sc{cx}} gate. These and other gates
  such as {\sc{cz}} and {\sc{cy}} are implemented using {\sc{cx}}
  gates in the native direction and flanked with single-qubit gates
  where needed. The core part of the layer that includes the {\sc{cx}}
  gates is twirled by randomly sampled Paulis $P_1P_2P_3P_4$ and their
  conjugation under the core gates, $P'_1P'_2P'_3P'_4$. Doing so for
  the layers in ({\bf{a}}) and ({\bf{b}}) gives implementations as
  shown in ({\bf{c}}) and ({\bf{d}}), respectively, with the noise
  associated with the core gates illustrated in red. Benchmarking of
  the layer fidelities is done using circuits ({\bf{e}}) where the
  depth-two block is repeated zero or more times for a given basis
  $B_1B_2B_3B_4$ and with randomly selected readout-mitigation gates
  $M_i$, which are sampled uniformly from identity and $X$ gates. The
  $S_i$ gates are used to control which Paulis are included in the
  fidelity pairs.}\label{Fig:Layers2}
\end{figure*}

\subsection{Full rankedness of $M$ when dealing with fidelity pairs}\label{Sec:FullRankM2}

Given a layer of non-overlapping two-qubit Clifford gates such that
each gate squares to the identity and such that the support of a Pauli
and that of the conjugation by the gate overlap (for instance,
conjugation of Pauli IX would not result in Pauli XI). This condition
is met for commonly used gates such that CX or CZ gates. We would like
to construct a Pauli-Lindblad noise model for the qubits that are
included in the gates, along with additional qubits for context, if
needed. The model terms consist of all unit-weight Paulis supported on
the model qubits, as well as all weight-two Paulis supported on pairs
of model qubits that are physically connected. We denote the complete
list of Pauli terms by $K$. Benchmarking using even number of layer
applications allows us to estimate the product of certain fidelity
pairs in a SPAM error free manner. Other fidelities can be estimated
based on the application of single layers, or based on symmetry
assumptions. In order to fit the noise model we assume access to the
following fidelity estimates of the Pauli noise channel: (1) for each
qubit $i \in [1,n]$ we have access to the fidelities for all
unit-weight Paulis $V_i = \{X_i,Y_i,Z_i\}$; and (2) for each connected
qubit pair $(i,j)$ we have access to products of fidelities for
$P_1 \in P_{i,j} = \{X_iX_j,X_iY_j,\ldots,Z_iZ_j\}$ and corresponding
Paulis $P_2 \in P_{i,j}'$ following application of the layer. We
assume that the Pauli terms on qubits $i$ and $j$ of $P_2$ are either
the same as those of $P_1$, or change to the identity. This can always
be achieved by inserting appropriate single-qubit gates during
benchmarking. For qubit pairs $(i,j)$ without a gate but with gates on
each of the qubits, $P_2$ can have a weight up to four. For pairs with
a gate the weight of $P_2$ is either one or two. The weight of Paulis
$P_1$ is always two. Collecting all $V_i$ and $P_1$ terms in list
$B_1$ and all $V_i$ and $P_2$ terms in list $B_2$ such that we have
the fidelity product for pairs at corresponding locations in the list
and setting the list of all model terms as $K=B_1$, we have the
following result.

\begin{theorem}
Given $B_1$, $B_2$ and $K$ as above, then $M = \mathcal{M}(B_1,K) + \mathcal{M}(B_2,K)$ is full rank. 
\end{theorem}
\begin{proof}
  We consider increasingly large blocks of $M$ and show that each of
  them is full rank. We start with the subblock corresponding to the
  unit-weight Paulis. We then add blocks corresponding to qubit pairs
  that contain a gate, and finally add the qubit pairs that do not
  contain a gate.  Starting with some notation, define by $I_n$ the
  $n\times n$ identity matrix and let
\[
Q = \left(\begin{array}{ccc}0&1&1\\1&0&1\\1&1&0\end{array}\right),
\quad
I = \left(\begin{array}{ccc}1&0&0\\0&1&0\\0&0&1\end{array}\right),
\quad
e = \left(\begin{array}{c}1\\1\\1\end{array}\right),
\quad
\mathbb{1} = Q+I = ee^T =
\left(\begin{array}{ccc}1&1&1\\1&1&1\\1&1&1\end{array}\right).
\]
For individual qubits it can be seen that $\mathcal{M}(V_i,V_j)$ is
$Q$ when $i=j$ and $0$ otherwise.  Setting
$V = V_1\oplus\cdots\oplus V_n$, where $\oplus$ denotes list
concatenation, it then follows that $\mathcal{M}(V,V) = I_n\otimes Q$,
which is full rank since both $I_n$ and $Q$ are full rank. Next, we
show that the matrix remains full rank if we add a single edge with a
gate. We illustrate this step on an example with three qubits and add
an edge on qubits (1,2). Consider
$\mathcal{M}(V\oplus P_{1,2},V\oplus P_{1,2})$, which has the
following structure
\begin{center}
\begin{tabular}{c|ccc|c}
&{\footnotesize{$V_1$}}&{\footnotesize{$V_2$}}&{\footnotesize{$V_3$}}
&{\footnotesize{$P_{1,2}$}}\\
\hline
{\footnotesize{$V_1$}} & $Q$ & $0$ & $0$ & $Q \otimes e^T$\\
{\footnotesize{$V_2$}} & $0$ & $Q$ & $0$ & $e^T \otimes Q$\\
{\footnotesize{$V_3$}} & $0$ & $0$ & $Q$ & 0\\
\hline
{\footnotesize{$P_{1,2}$}} & $Q\otimes e$ & $e\otimes Q$&0
& $I\otimes Q + Q\otimes I$\\
\end{tabular}
\end{center}
We can eliminate the $\mathcal{M}(P_{1,2},V)$ block by subtracting the
Kronecker product of row-block for $V_1$ by $e$ and the Kronecker
product of $e$ with the row-block for $V_2$. Doing so changes to
lower-right $\mathcal{M}(P_{1,2},P_{1,2})$ block to
\[
I\otimes Q + Q\otimes I - Q\otimes \mathbb{1} - \mathbb{1}\otimes Q =
I\otimes Q + Q\otimes I - Q\otimes (Q+I) - (Q+I)\otimes Q
= -2Q\otimes Q,
\]
which is full rank. This means that $M_{1,2}$ is full rank. Now
consider $M'_{1,2} = \mathcal{M}(V\oplus P_{1,2}',V\oplus E_{1,2})$ in
which we replaced $P_{1,2}$ by $P'_{1,2}$ in the rows. The rows in
$M'_{1,2}$ corresponding to elements in $E'_{1,2}$ that are weight two
exactly match those in $M_{1,2}$. The remaining rows correspond to
Paulis with weight one and therefore correspond to one of the rows in
$V$. Elimination of the lower-left block therefore results in a
lower-right block that is $-2Q\otimes Q$, but with some rows zeroed
out. The rows we sweep with in $M_{1,2}$ and $M'_{1,2}$ are identical,
which means we can perform the row sweeps with half the weight in the
sum $M_{1,2}+M'_{1,2}$. The resulting matrix is $D(Q\otimes Q)$ with
diagonal matrix $D$ with terms $-2$ and $-4$. That means that even
though $M'_{1,2}$ may not be full rank, the sum of the two matrices
is. Given that the qubit pairs with gates do not have any overlap we
can simply repeat the same procedure for each such pair.
Moving on to pairs $(i,j)$ without a gate we note that we can again
factor each Pauli term as a product of two Paulis. If there is a gate
$(i,i')$ then one part of the factorization will be a Pauli supported
on either $i$ or $(i,i')$. If there is no gate on $i$, then the Pauli
is simply supported on $i$. The same applies to qubit $j$ with
possible gate $(j,j')$. Given that gates do not overlap we never have
$i'=j'$ and the supports of the two factors will therefore always be
disjunct. Based on the assumptions we have that corresponding Paulis
in $P_{i,j}$ and $P_{i,j}'$ have the same term for qubit $i$ and
likewise for qubit $j$. That means that
$\mathcal{M}(P_{i,j},P_{i,j}) = \mathcal{M}(P_{i,j}',P_{i,j}) =
I\otimes Q+Q\otimes I$. If qubit $i$ does not have a gate we sweep the
first Pauli factors with a row from $V_i$. If qubit $i$ does have an
incident gate we can sweep with the appropriate row from the
$P_{i,i'}$ (or $P_{i',i}$) block if the support changes, and otherwise
use a row from $V_i$. Doing the same for $j$, we see that we can sweep
the lower-right block of the new matrix and end up with a combined
$-4Q\otimes Q$ lower-right block. As an aside, note that sweeping is
done directly on the combined matrix, since all Paulis $P_{i,j}$ on
pairs $(i,j)$ with a gate have weight two, whereas we possibly need to
sweep with their weight-one counterpart found only in $P_{i,j}'$.
\end{proof}

\section{Probabilistic error cancellation and error-analysis}
The purpose of noise mitigation is to accurately estimate the
expectation value of observables. For a circuit consisting of ideal
operations $\mathcal{U}_l\circ\cdots\circ\mathcal{U}_1$, initial state
$\rho_0$, and observable $A$, which we assume to have an operator norm
$\norm{A} \leq 1$, we would like to estimate
\[
\langle A\rangle = \Tr\left[A \, \mathcal{U}_l\circ\cdots\circ\mathcal{U}_1(\rho_0)\right].
\]
Each of the maps $\mathcal{U}_i$ is available only through its noisy
version $\tilde{\mathcal{U}}_i = \mathcal{U}_i \circ \Lambda_i$, where
$\Lambda_i$ is twirled and assumed to be a Pauli-Lindbladian
channel. Using the techniques described earlier, we can learn this
channel in experiment up to an error as given in Theorem
\ref{Thm:Channel} giving rise to the channel estimate
$\hat{\Lambda}_i$. We can implement the inverse $\hat{\Lambda}_i^{-1}$
of this channel estimate in experiment as described in section
\ref{sec:Sample_inv}. \\

\subsection{Sampling from the inverse}\label{sec:Sample_inv}
The PEC error mitigation protocol asks that we sample the noise
inverse by a quasi-probabilistic technique described in
Section~\ref{Sec:QPNoiseInversion}. For the noise process we are
working with the Pauli-twirling method as explained in
Section~\ref{Sec:Twirling} and learn the resulting sparse noise model
following Section~\ref{Sec:Benchmarking}. Although our noise
model~\eqref{Eq:FinalForm} represents a Pauli channel, it is not in
the canonical form shown in~\eqref{Eq:PauliChannel}. If we denote by
$\mathcal{K}$ the set of $k$ values that are included in the noise
model, then it is easily seen that there there are
$2^{\vert\mathcal{K}\vert}$ different products of the identity and
$P_k$ terms in~\eqref{Eq:FinalForm}, each with a possibly different
weight. In order to find the coefficient for a certain Pauli $P$ in
the canonical representation~\eqref{Eq:PauliChannel} we would have to
identify and sum up weights of all products that result in this
particular Pauli to obtain the right coefficient in the canonical
expansion. Moreover, the error-mitigation method asks that we then
invert and re-normalized the expansion accordingly. Following these
steps as outlined directly is computationally clearly intractable.\\

Instead, we produce the samples from the inverse by exploiting the
product structure of the model~\eqref{Eq:FinalForm}. The channel
$\Lambda = \mbox{exp}[{\cal L}] $ is given as a product of
$\vert\mathcal{K}\vert$ individual (commuting),
c.f.~\eqref{Eq:SuperLindbladian2}, Pauli channels
$\left(w_k\rho + (1-w_k)P_k\rho P_k\right)$, with
$w_k = (1+e^{-2\lambda_k})/2$. The inverse of the overall channel then
reduces to the product of the individual inverse channels. We can
write these inverse channels as
$(2w_k - 1)^{-1}\left(w_k\rho - (1-w_k)P_k\rho P_k\right)$. The full
inverse channel is given then by the product
\begin{align}
    \Lambda^{-1}(\rho) = \gamma \prod_{k \in {\cal K}}\left(w_k\cdot - (1-w_k)P_k\cdot P_k\right)\rho,
\end{align}
where the sampling overhead $\gamma$ is given as the product of the
individual normalizing factors so that
\begin{equation}\label{Eq:Eta}
\gamma =  \prod_{k\in\mathcal{K}} (2 w_k-1)^{-1} = \exp\Big(\sum_{k\in\mathcal{K}}2\lambda_k\Big)\;.
\end{equation}
This means the application of the inverse $\Lambda^{-1}(\rho)$ can be
sampled according to the following steps. For every $k \in {\cal K}$
we sample the identity matrix with probability $w_k$, and $P_k$ with
probability $1-w_k$. Each time we sample a Pauli matrix $P_k$, we
record the minus sign $(-1)$. To produce a single sample of the full
inverse it then suffices to multiply all the (Abelian) Pauli terms we
have sampled as well as all observed signs. The final Pauli is then
inserted in the random circuit instance and the measurement sample for
this instance is then obtained by multiplying the observed outcome
with the final sign and the factor $\gamma$. This procedure has to be
applied at every layer $i = 1,2, \ldots l$ of the circuit, c.f. Fig~1a
(main text) so that all these factors compound. This means that the
sampling protocol has to be applied to the noise channel $\Lambda_i$
for each layer $i = 1,2,\ldots,l$. This means that every layer
contributes a multiplicative factor of $\gamma_i$ to the sampling
overhead resulting in the full overhead
$\gamma(l) = \prod_{i=1}^l \gamma_i$. Likewise, we have to record the
total number of times $m$ by which we have sampled a Pauli matrix for
all the layers, so that we can assign the global sign flip as
$(-1)^m$. Note, that this sampling procedure does not change the form
of the random quantum circuits we need to sample. In fact this error
mitigation procedure only uses instances of Pauli-twirled quantum
circuits and only modifies the classical distribution from which the
circuits are drawn and multiples the output by the factor
$(-1)^m\gamma(l)$. These additional steps are all taken only in
classical pre- and post-processing.

It is also possible to explicitly expand subsets of terms
in~\eqref{Eq:FinalForm} and work with Pauli channels that contain more
terms. Since combining terms we are able to decrease $\gamma$. This
enables us to make a trade-off between the computational complexity of
expanding the channels and sample complexity due to scaling parameter
$\gamma$.

\subsection{Error bounds for probabilistic error cancellation}
\label{Sec:ErrorAnalysisObservables}

Let us assume for simplicity that observable $A$ can be diagonalized
in the computational basis and has eigenvalues $X \in \{-1,+1\}$, as
is for example the case for Pauli observables. We absorb the factor
$\pm 1$ that originate from the quasi-probability sampling method,
c.f. section \ref{sec:Sample_inv} into the random variable $X$
already. Note, that the general case can be reduced to estimating
Pauli-observables or other binary measurements. Furthermore, while
considering the error bound for the PEC protocol, we assume that there
are no state preparation and readout errors. These can be addressed
through other means \cite{BER2020MTa-arXiv,BRA2021SKMa}. This means,
we can sample $N$ noise-mitigated circuit instances and measure the
observable $A$ to obtain $r = 1,2 \ldots N$ individual samples
$X_{r}$.  From these, we can estimate the observable expectation value
as
\begin{equation}\label{Eq:AHatN}
\langle \hat{A}_N\rangle := \gamma(l)\frac{1}{N}\sum_{r=1}^N X_{r} = \gamma(l)\mathbb{E}(X).
\end{equation}
The following Theorem provides a bound on the difference between the
actual and estimated expectation value for observable $A$.

\begin{theorem}\label{thm:error-bound}
  Assume that all noise channels $\Lambda_i$ are learned at each layer
  $i = 1,2, \ldots l$ of the circuit with a multiplicative error as in
  Theorem~\ref{Thm:Channel}. Then it holds with probability at least
  $1-\delta$ for $\delta > 0$, that
\[
\vert\langle{A}\rangle - \langle \hat{A}_N\rangle\vert
\leq 
(C_{\epsilon}^{l \tau} - 1) + \gamma(l) \sqrt{2\log(2/\delta)/N},
\]
where $N$ is the number of error-mitigation circuit instances,
$\gamma(l) = \prod_{i=1}^l \gamma_i$ is the product of the scaling
factors $\gamma_i$ for the estimated channels $\hat{\Lambda}_i$, and
$C_{\epsilon}$ and $\tau$ are as in Theorem~\ref{Thm:Channel}.
\end{theorem}

To simplify notation throughout the manuscript we have simply referred
to the noise channel as $\Lambda$ independently of whether we are
dealing with the ideal channel or its estimate $\hat{\Lambda}$
obtained from the noise-learning procedure. To account for a full
error analysis we now have to make an explicit distinction. Note
however, that crucially both the error-bound in the theorem
\ref{thm:error-bound}, as well as the quasi-probabilistic noise
inversion method in section \ref{sec:Sample_inv} depend on the
estimated value for $\gamma(l)$ obtained from the learning experiments
and do not need the knowledge of the ideal values for the exact
channel $\Lambda$. Furthermore, we point out that the estimates can
naturally be related to the ideal values by Theorem~\ref{Thm:Channel}.

\begin{proof}
  There are two contributions to the error, first the increased
  sampling error that arises due to the PEC protocol itself and second
  the error we occur due to errors in the noise-learning procedure
  that determine the estimate for the $\hat{\Lambda}_i$. As discussed
  in this section, the random variable $X$ in Eq.~\eqref{Eq:AHatN}
  satisfies
\begin{equation}\label{Eq:EtaExpX}
\gamma(l)\mathbb{E}(X) =
\Tr\left [ A \, \left( \mathcal{U}_l \circ\Lambda_l \circ \hat{\Lambda}_l^{-1} \right)
\circ\cdots\circ
\left(\mathcal{U}_1 \circ \Lambda_1 \circ \hat{\Lambda}_1^{-1}\right)(\rho_0)\right].
\end{equation}
 Bounding the right-hand side of Hoeffding's inequality~\eqref{Eq:Hoeffding} by $\delta$ gives an additive error for $\mathbb{E}(X)$ of
\begin{equation}\label{Eq:EpsilonSample}
\epsilon_s \leq \sqrt{2\log(2/\delta)/N}
\end{equation}
with probability at least $1 - \delta$. This allows us to estimate
$\mathbb{E}(X)$ in Eq.~\eqref{Eq:EtaExpX} up to an additive sampling
error of $\epsilon_{s}$. In order to bound
$\vert\langle A\rangle - \langle\hat{A}\rangle\vert$, we first define
\[
\mathcal{T}_k = \mathcal{U}_k\circ\cdots\circ\mathcal{U}_1,
\qquad\mbox{and}\qquad
\mathcal{S}_k = \left(\mathcal{U}_k \circ \Lambda_k \circ\hat{\Lambda}_k^{-1}\right) \circ\cdots\circ
\left(\mathcal{U}_1 \circ \Lambda_1 \circ \hat{\Lambda}_1^{-1}\right).
\]
It then follows from the triangle inequality and properties of the trace that
\begin{align}
\vert\langle A\rangle - \langle\hat{A}\rangle\vert
& \leq \gamma(l) \epsilon_s + \vert \langle A\rangle - \gamma(l)\mathbb{E}(X)\vert\notag\\
& = \gamma(l) \epsilon_s +
\vert \Tr\left[A\,\mathcal{T}_l(\rho_0)\right] - \Tr\left[A\,\mathcal{S}_l(\rho_0)\right]\vert\notag\\
& \leq \gamma(l) \epsilon_s +
\Vert A\Vert
\Vert (\mathcal{T}_l - \mathcal{S}_l)(\rho_0)\Vert_1\notag\\
& \leq \gamma(l) \epsilon_s +
\dnorm{\mathcal{T}_l - \mathcal{S}_l},\label{Eq:AAhat}
\end{align}
with $\norm{A} \leq 1$. The last inequality follows from the
definition of the diamond norm $\dnorm{\cdot}$, which has a number of
useful, properties.

For TCP-maps $T$ we have $\dnorm{T} \leq 1$, whereas for general
linear maps $A$ and $B$ the norm is sub-multiplicative and thus
satisfies $\dnorm{A \circ B} \leq \dnorm{A}\dnorm{B}$. For linear maps
we therefore have
\begin{equation}\label{Eq:DiamonNormAB}
\dnorm{A_1 \circ A_2 - B_1 \circ B_2}
\leq
\dnorm{A_1 - B_1}\dnorm{A_2} + \dnorm{B_1}\dnorm{A_2 - B_2}.
\end{equation}
Note, that both $\Lambda$ and $\hat{\Lambda}^{-1}$ have diagonal
Pauli-transfer matrices. The combined map
$\Lambda \circ \hat{\Lambda}^{-1}$ has therefore eigenvalues
$f_j(\widehat{f_j^{-1}})$ that are bounded by
$C_{\epsilon}^{-\tau} \leq f_j(\widehat{f_j^{-1}})\leq
C_{\epsilon}^{\tau}$ according to Theorem~\ref{Thm:Channel}. Hence, we
immediately have that
$C_{\epsilon}^{-\tau} \leq \dnorm{\Lambda \circ \hat{\Lambda}^{-1}}
\leq C_{\epsilon}^{\tau}$.  From this it follows that
\begin{align*}
\dnorm{\mathcal{T}_l - \mathcal{S}_l}
& = \dnorm{\mathcal{U}_l\circ\mathcal{I}d\circ \mathcal{T}_{l-1}
- \mathcal{U}_l \circ (\Lambda_l \circ \hat{\Lambda}_l^{-1})\circ\mathcal{S}_{l-1}}\\
& \leq 
\dnorm{\mathcal{I}d\circ \mathcal{T}_{l-1}
- (\Lambda_l \circ \hat{\Lambda}_l^{-1})\circ\mathcal{S}_{l-1}}\\
& \leq
\dnorm{\mathcal{I}d - \Lambda_l \circ \hat{\Lambda}_l^{-1}}
\dnorm{\mathcal{T}_{l-1}}
+ \dnorm{\Lambda_l \circ \hat{\Lambda}_l^{-1}}
\dnorm{\mathcal{T}_{l-1}-\mathcal{S}_{l-1}}\\
& \leq (C_{\epsilon}^{\tau} - 1) +
C_{\epsilon}^{\tau}\dnorm{\mathcal{T}_{l-1}-\mathcal{S}_{l-1}}.
\end{align*}
For the final iteration we can take $\mathcal{T}_{-1}$ and
$\mathcal{S}_{-1}$ to be the identity, giving
$\dnorm{\mathcal{T}_1 - \mathcal{S}_1} \leq
C_{\epsilon}^{\tau}-1$. Solving the resulting recurrence relation
gives
\[
\dnorm{\mathcal{T}_l - \mathcal{S}_l} \leq C_{\epsilon}^{l \tau} - 1.
\]
The result then follows from Eqs.~\eqref{Eq:EpsilonSample}
and~\eqref{Eq:AAhat}.
\end{proof} 

\subsection{Weak exponential scaling}\label{Sec:Scaling}

Consider a layer consisting of $k$ non-overlapping two-qubit gates
such that each gate on qubits $i$, $j$ is affected by a local
two-qubit depolarizing channel
$\mathcal{D}(\rho) = f\rho + \frac{1-f}{4}\mbox{Tr}_{ij}[\rho]$, such
that the fidelity for any Pauli is $f$. For each channel we can form a
two-local error model, for which it follows from~\eqref{Eq:Mf} that
all model coefficient in $\lambda$ are $-\log(f)/16$.  Given that the
gates do not overlap we can combine the individual noise channel into
the layer-level noise channel using the results from
Section~\ref{Sec:ChannelOperations}. It is then easy to see that the
overall noise model has $15k$ nonzero model coefficients, all equal to
$-\log(f)/16$.  Using~\eqref{Eq:Eta} it then follows that
$\gamma = \exp(-(15k/8)\log(f))$. This expression allows us to analyze
the growth of $\gamma$ for the Ising model in the main text. For $n$
qubits we have one layer with $\lfloor n/2\rfloor$ gates and one layer
with $\lfloor (n-1)/2\rfloor$ gates. In Figure~\ref{Fig:GammaEpsilon}a
we plot the value of $\gamma$ as a function of $(1-f)$ for different
number of qubits $n$. The plot in Figure~\ref{Fig:GammaEpsilon}b then
shows for $n=50$ the relative number of circuit instances that need to
be sampled to attain a similar variance in the estimated observable
for different number of Trotter steps. Although the curves rise
quickly in the error $(1-f)$, the opposite is also true: minor
improvements in gate fidelities lead to a huge decrease in the number
of circuit instances that need to be sampled and therefore enable
simulation of larger systems.

\begin{figure}[ht]
\centering
\begin{tabular}{ccc}
\includegraphics[width=0.45\textwidth]{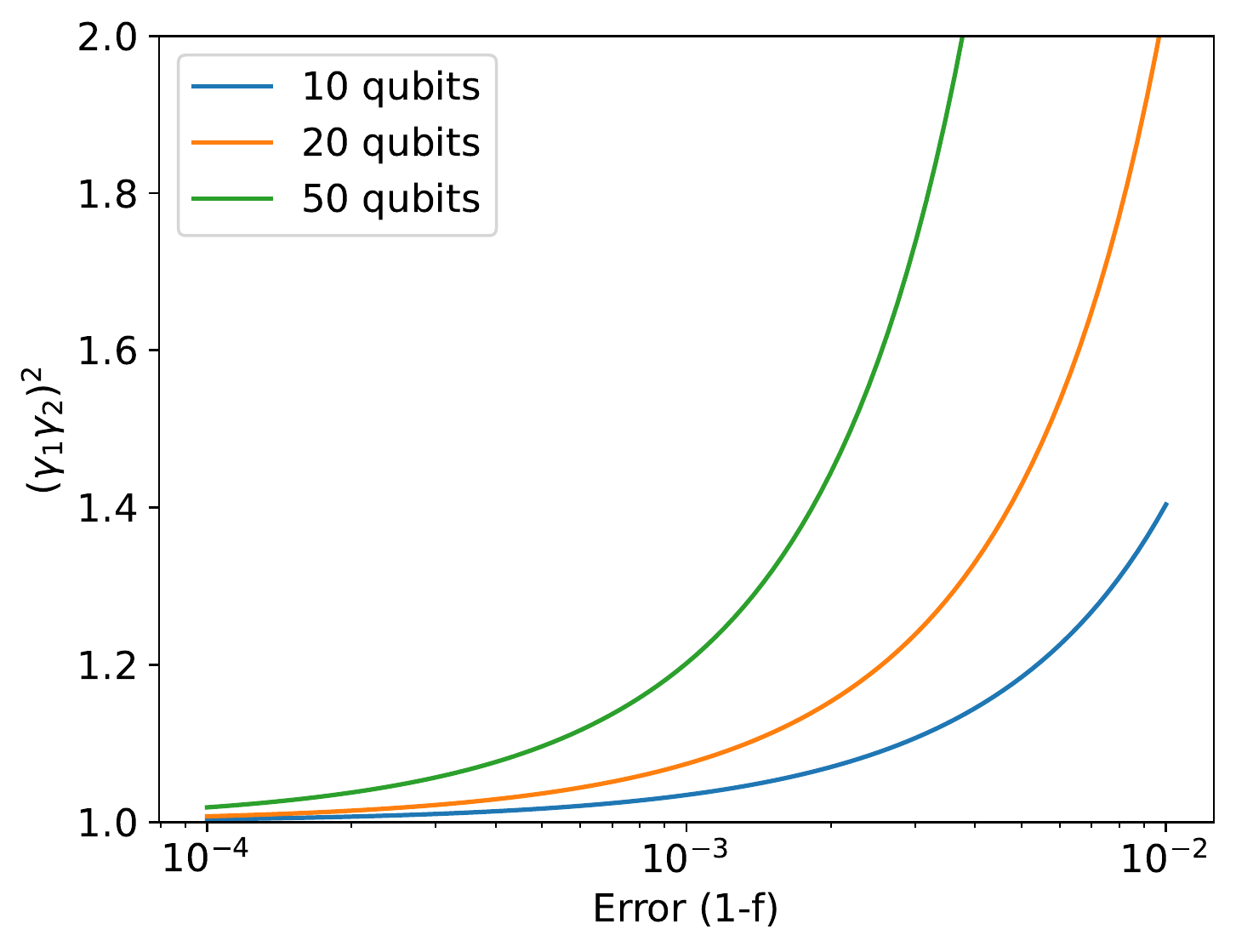}&&
\includegraphics[width=0.45\textwidth]{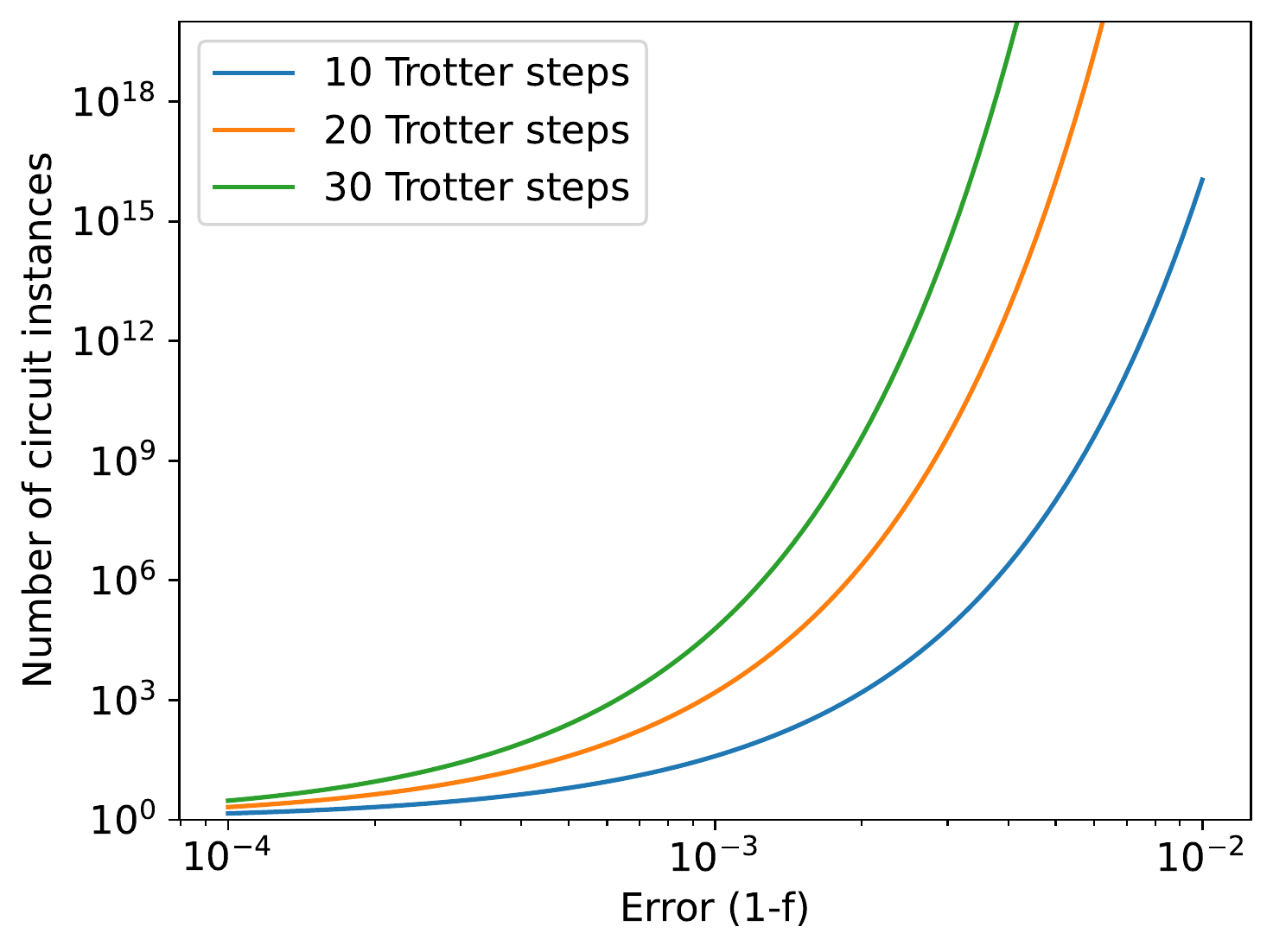}\\[-2pt]
{\bf{a}} && {\bf{b}}
\end{tabular}
\caption{Scaling of ({\bf{a}}) the sampling overhead factor
  $\gamma = (\gamma_1\gamma_2)^2$ for the Ising simulations discussed
  in the main text, in case each {\sc{cx}} gate is affected by
  depolarizing noise with fidelity $f$. Plot ({\bf{b}}) shows the
  relative number of circuit instances needed for $n=50$ to achieve a
  similar variance in estimates. Slight improvements in the gate
  fidelity lead to significant reductions in the required number of
  circuit instances.}\label{Fig:GammaEpsilon}
\end{figure}

\section{Setup of the experiment}\label{sec:exp-setup}

\subsection{Devices of the experiment}\label{subsec:Devices-of-exp}

We performed the experiments on superconducting quantum processors
\cite{Jurcevic2020,Zhang2020-laser-anneal}, which utilized
fixed-frequency transmon qubits \cite{Koch2007}. All devices were
patterned to realize heavy-hexagon lattices (see Fig.~1d of the main
text). The experiments presented in the main text were all obtained on
the same 27-qubit Falcon processor, named~\textsc{ibm\_hanoi}.
Other iterations of the protocol were executed on other Falcon chips
(\textsc{ibm\_mumbai}, \textsc{ibm\_kolkata}, \textsc{ibm\_syndey},
and \textsc{ibm\_montreal}), with results and conclusions similar to
those presented for \textsc{ibm\_hanoi}. We view these additional
tests as an indicator for the reproducibility and robustness of the
protocol.

\subsection{Specifications of the primary device}\label{subsec:Specifications-of-exp}

Basis gates --- All circuits were transpiled to the standard basis
gate set~$\big\{ I,\sqrt{X},X,R_{Z},\mathrm{CX}\big\}$. The
single-qubit gates $\sqrt{X}$ and $X$ were implemented using the
standard circuit quantum electrodynamics (cQED) all-microwave-control
setup, using Gaussian pulses with calibrated DRAG
decoupling~\cite{Motzoi2009,JChow2010-DRAG}, each with a total gate
time of 35.5~ns (4$\sigma$ Gaussian pulses). The~$I$ and~$R_{Z}$
pulses were virtualized~\cite{PhysRevA.96.022330}, and hence took no
time. The two-qubit CX gates were implemented using cross-resonance
pulses \cite{Paraoanu2006,Chow2011} with gate times optimized to
maximize fidelity. Figure~\ref{fig:Total-two-qubit-gate} shows the
chip topology of \textsc{ibm\_hanoi} along with the duration of its CX
gates in their native direction. The CX gate duration ranges from 181
to 519~ns, and their average error, as estimated by randomized
benchmarking, was~0.98\%. In all experiments, the qubits were cooled
to the ground state prior to the start of the protocol.

\begin{figure}[!t]
\begin{centering}
\includegraphics[width=0.90\textwidth]{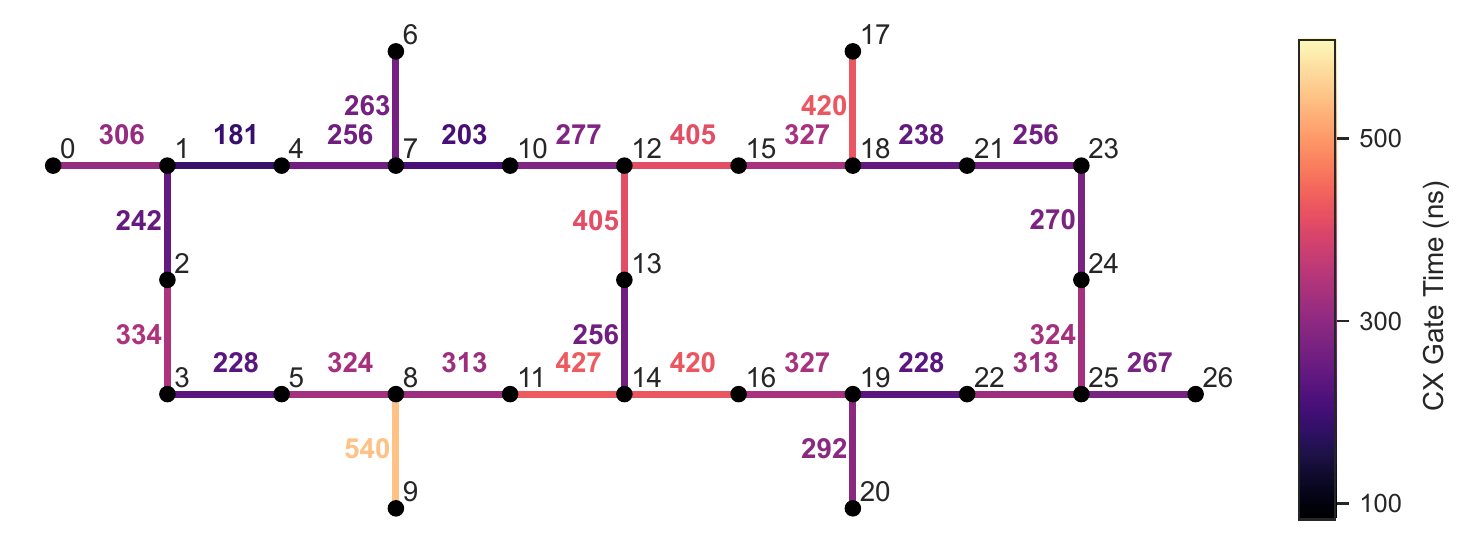}
\par\end{centering}
\caption{
\label{fig:Total-two-qubit-gate}\textbf{Quantum
processor topology and native CX gate duration for {\sc{ibm\_hanoi}}.}
Each node represents a qubit and is labeled by the physical qubit
number (black text). CX Gates are represented by edges, along with
their gate time in nanoseconds (colored text).  }
\end{figure}

\medskip\noindent{Coherence} --- Processor \textsc{ibm\_hanoi} had a
quantum volume of 64 \cite{Jurcevic2020} and average
energy-relaxation~$T_{1}$ and Hanh echo~$T_{2}^{E}$ times of
$151\,\mathrm{\mu s}$ and $107\,\mathrm{\mu s}$, respectively. As
typical for supercomputing qubits~\cite{Klimov2018,Carroll2021}, these
times fluctuated over the duration of the experiments. In
Fig.~\ref{fig:t2e-violin}, we summarize the distribution of the
variations in $T_{2}^{E}$ over a two-month long period for each of the
27 qubits. We limited the effects of temporal fluctuations by
interleaving our mitigation experiments with noise-learning runs every
few hours.

\medskip\noindent{Readout} --- We define the readout assignment
fidelity per qubit as
$\mathcal{F}_{a}=1-{\textstyle\frac{1}{2}}\left(P(1\,\vert\,0)+P(0\,\vert\,1)\right),$
where $P(A\,\vert\,B)$ is the empirical probability to measure the
qubit in state $A\in\{0, 1\} $ given that the qubit was nominally
prepared in state~$B\in\left\{ 0,1\right\} $.  The average assignment
readout error $1-\mathcal{F}_{a}$ across all qubits in our device was
$2.5\%$. We note that the probability distribution is biased due to
energy relaxation such that $P(1\,\vert\,0)<P(0\,\vert\,1)$.

\begin{figure}[t]
\begin{centering}
\includegraphics{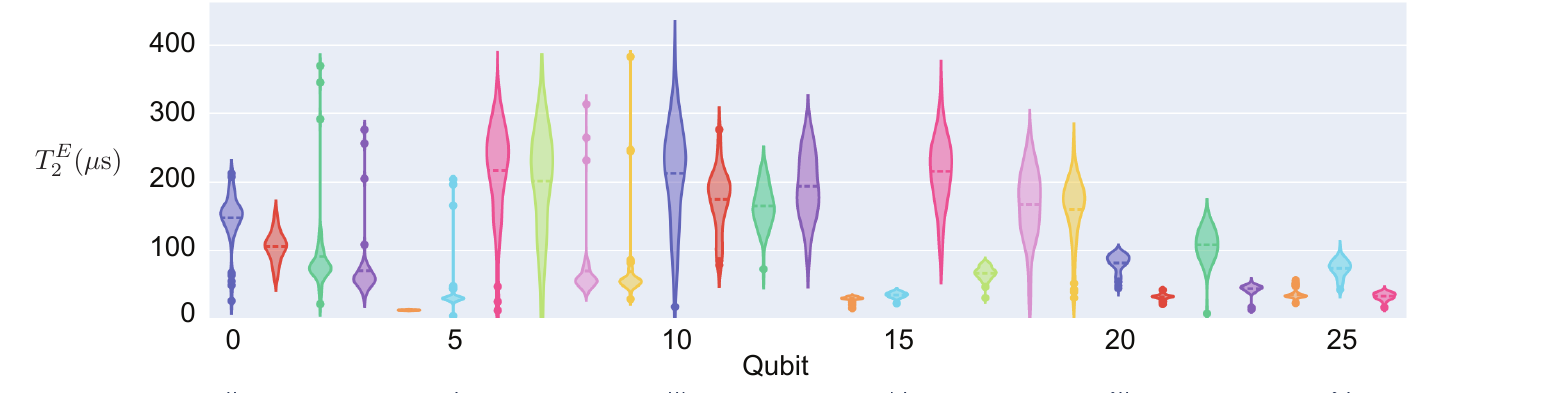}
\par\end{centering}
\caption{ 
\label{fig:t2e-violin}\textbf{Device
coherences and time variability. }Distribution of the $T_{2}$-echo
relaxation times for each of the qubits of device \textsc{ibm\_hanoi}
over a two-month period, depicted as a violin plot. A larger
horizontal width of the violin plot indicates a larger probability of
sampling this $T_{2}^{E}$ time. The horizontal dashed lines indicate
the mean of the distribution.}
\end{figure}

\begin{table}[th]
\begin{centering}
{\setlength{\extrarowheight}{2pt}
\setlength{\doublerulesep}{4pt}
\addtolength{\tabcolsep}{4pt}
\begin{tabular}{l|c|c}
 & Symbol & Value\tabularnewline
\hline 
Quantum volume & $\mathrm{QV}$ & $64$\tabularnewline
Energy relaxation lifetime & $T_{1}$ & $151\,\mathrm{\mu s}$\tabularnewline
Hanh echo time & $T_{2}^{E}$ & $107\,\mathrm{\mu s}$\tabularnewline
Readout assignment error & $1-\mathcal{F}_{a}$ & $2.5\%$\tabularnewline
CX error &  & 0.98\%\tabularnewline
\end{tabular}}
\par\end{centering}
\caption{\label{tab:Summary-table-device-spec}Summary table of average device
metrics for \textsc{ibm\_hanoi}. Symbols explained in Sec.~\ref{subsec:Specifications-of-exp}.}
\end{table}

\subsection{Dynamical decoupling}\label{subsec:Dynamical-decoupling}
\label{subsec:dd-and-learning}

As a result of the different gate times, qubits in our layer of CX
gates can experience idle periods. This holds especially true for
context qubits, which are idle for the full duration of the layer. To
lessen the effects of decoherence and low-frequency noise during these
idle periods, we apply dynamical decoupling
(DD)~\cite{Viola1999,Zanardi1999}, which was recently demonstrated to
improve circuit fidelity \cite{Jurcevic2020}. For dynamical decoupling
we use a standard~$X_{p}-X_{m}$ sequence, which is the simplest
version of a Car-Purcell-Meiboom-Gill (CPMG) echo
train~\cite{Carr1954,Meiboom1958}. This is illustrated in
Fig.~\ref{fig:Structure-of-dd}.

\begin{figure}[ht]
\begin{centering}
\includegraphics[width=282pt]{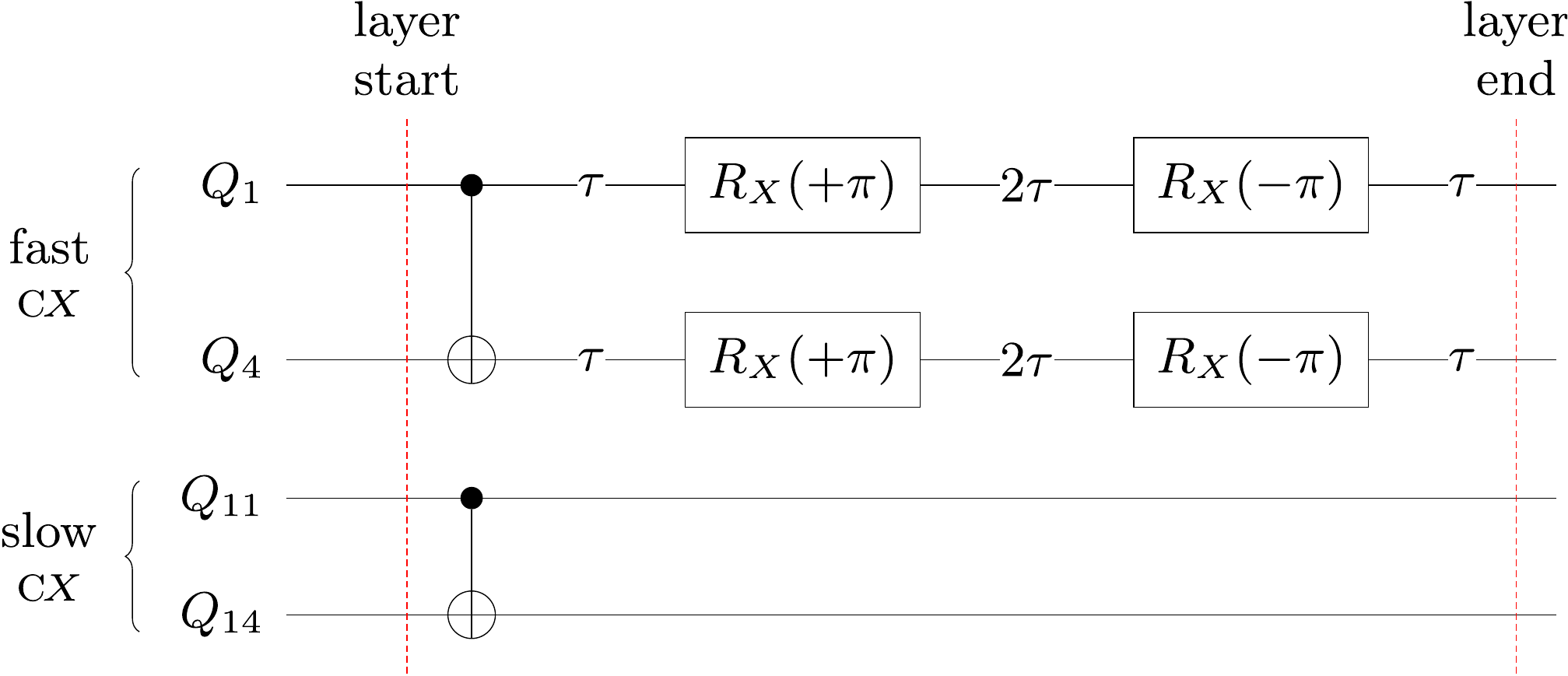}
\par\end{centering}
\caption{
\label{fig:Structure-of-dd}
\textbf{Dynamical decoupling inside a layer.} Structure of a simple
dynamical decoupling sequence used during idle qubit times,
illustrated on a four-qubit layer with concurrent \textsc{cx} gates on
qubits $1-4$ and $11-14$. The delay duration~$\tau$ is calculated as
the idle time minus the duration of the two $R_X$ gates, divided by
four.}
\end{figure}

\begin{figure}[th]
\begin{centering}
\includegraphics[width=\textwidth]{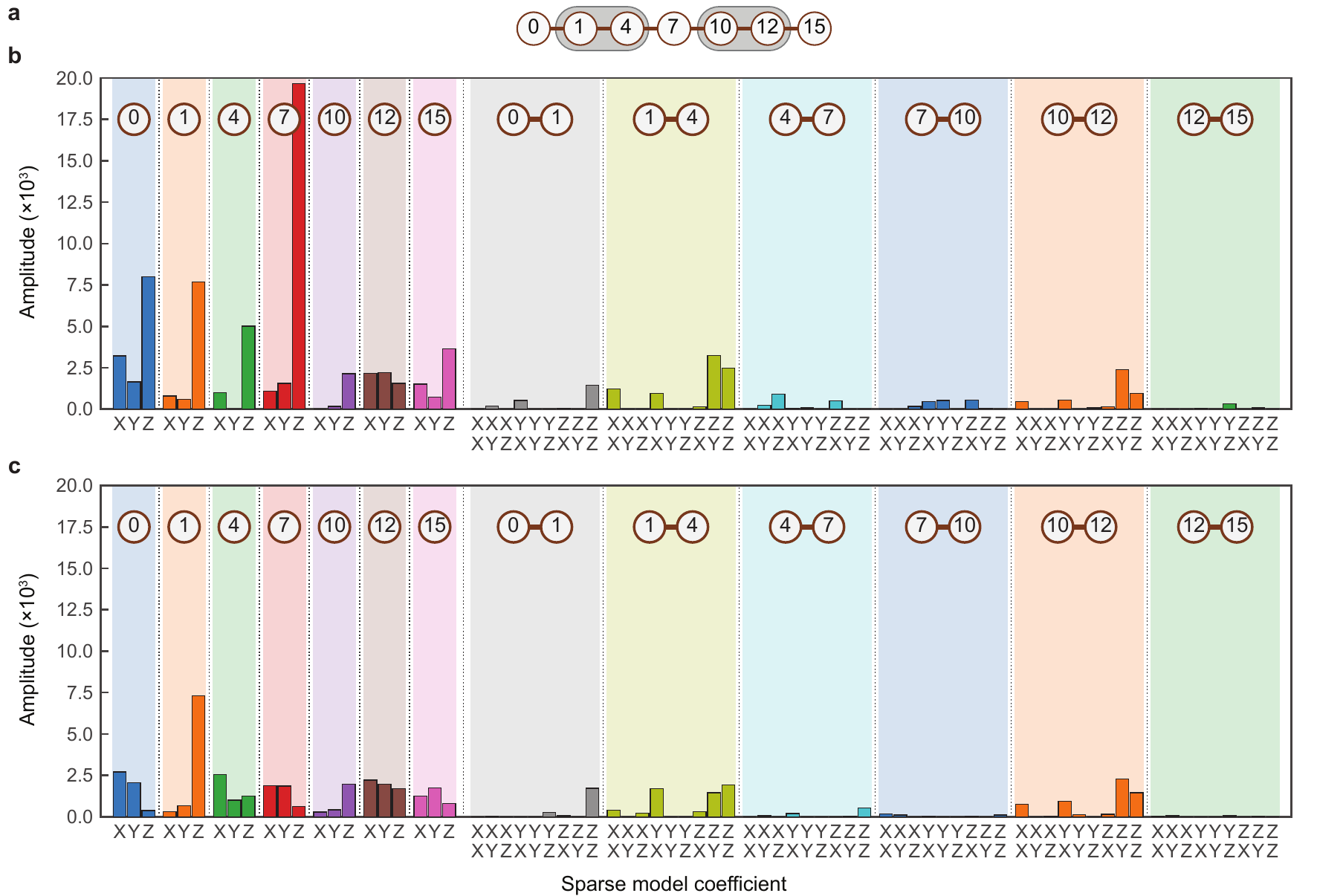}
\par\end{centering}
\caption{
\label{fig:dd-learning}\textbf{Noise model coefficients in the presence and absence of dynamical decoupling.} (\textbf{a})
Schematic depiction of a 7-qubit layer on \textsc{ibm\_kolkata} with
{\sc{cx}} gates on qubits $1-4$ and~$10-12$, and (idle) context qubits
0, 7, and 15. ({\bf{b}},{\bf{c}}) Plots of the learned noise-model
coefficients in the absence ({\bf{b}}) and presence ({\bf{c}}) of
dynamical decoupling within the layer using an $X_{p}-X_{m}$ sequence
on the context qubits. The numbered circles overlaid over each section
of the bar plots indicate the support of the model Pauli terms.}
\end{figure}

To study the effect of dynamical decoupling on the structure of noise
in our system, we considered a 7-qubit layer with two CX gates and
three (idle) context qubits, as illustrated in
Fig.~\ref{fig:dd-learning}a.  We then learned the layer with and
without dynamical decoupling applied to the context qubits. The noise
model coefficients obtained without dynamic decoupling are shown in
Fig.~\ref{fig:dd-learning}b. The dominant noise in the system
corresponds to unit-weight Pauli-Z terms.  The origin of these
dominant noise terms may be attributed to $T_{2}$ qubit dephasing and
other coherent $Z$-noise arising from crosstalk.  Repeating the
experiment with dynamical decoupling enabled resulted in the model
coefficients shown in Fig.~\ref{fig:dd-learning}c.  It is seen that
the large Pauli-Z noise terms on the idle qubits are significantly
reduced.

\subsection{Additional learning and control experiments}\label{Sec:AdditionalLearning}

\begin{figure}[th]
\begin{centering}
\includegraphics[width=0.99\textwidth]{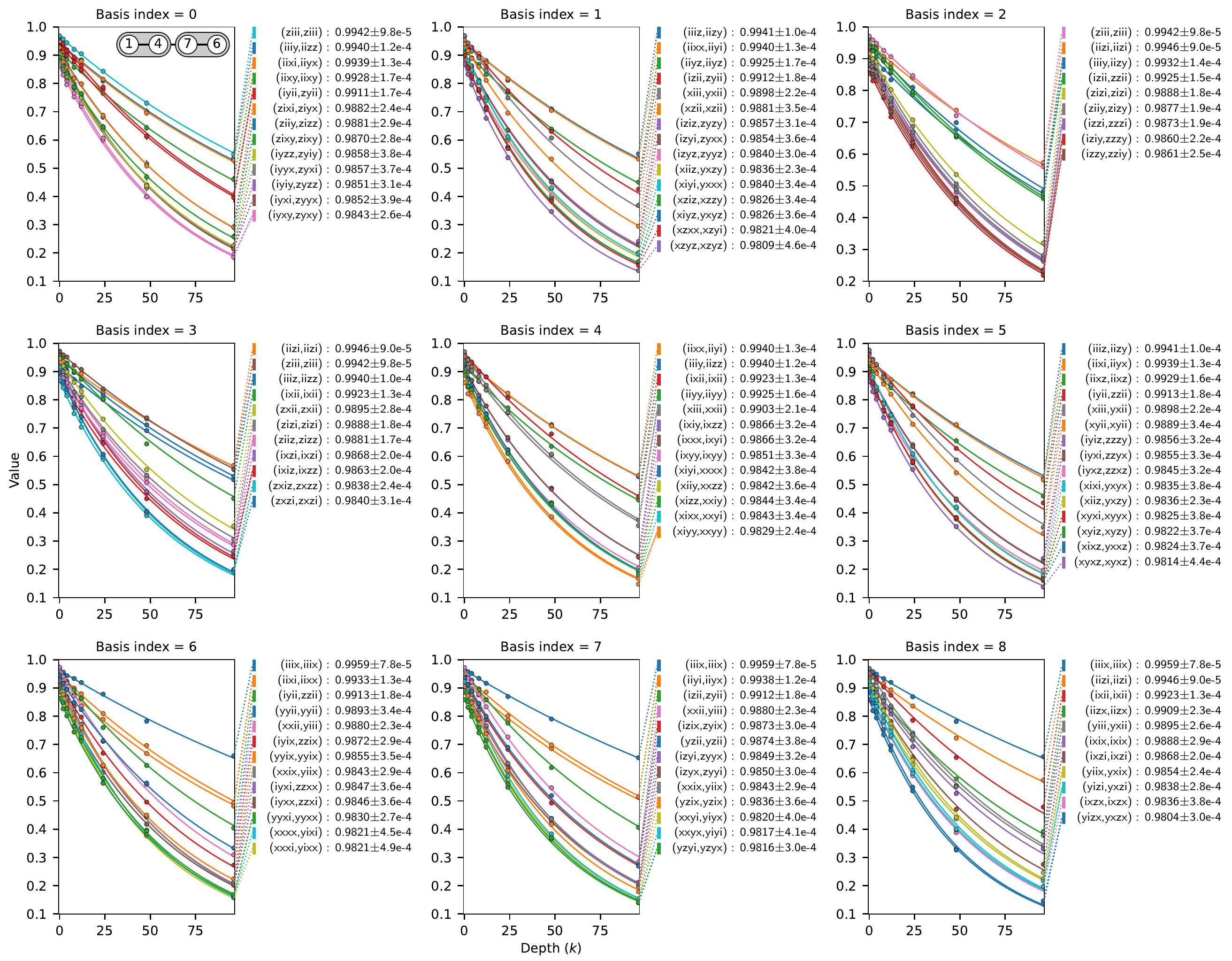}
\end{centering}
\caption{\textbf{Raw
data of noise-learning experiments.} 
Plots of the learning data in all nine different bases for the
four-qubit layer of Fig.~2a of the main text. The layer consists of
two concurrent \textsc{cx} gates applied to qubits 1--4 and 7--6 on
\textsc{ibm\_hanoi} (see top-left inset).  Each plot gives the
observable expectation values for different circuit depths along with
the exponentially decaying function fitted through the data points,
with decay rate corresponding to the square root of the product of two
fidelities. Fitting for a given fidelity pair is done for all
occurrences within and across the different bases. For instance the
fidelity for {\sc{ziii}} is determined based on the data obtained for
basis indices 0 and 2, whereas the fidelity pair
{\sc{iyii}}--{\sc{zyii}} occurs twice for basis index
0.}\label{fig:dd-learning-1}
\end{figure}

\begin{figure}[th]
\begin{centering}
\includegraphics[width=\textwidth]{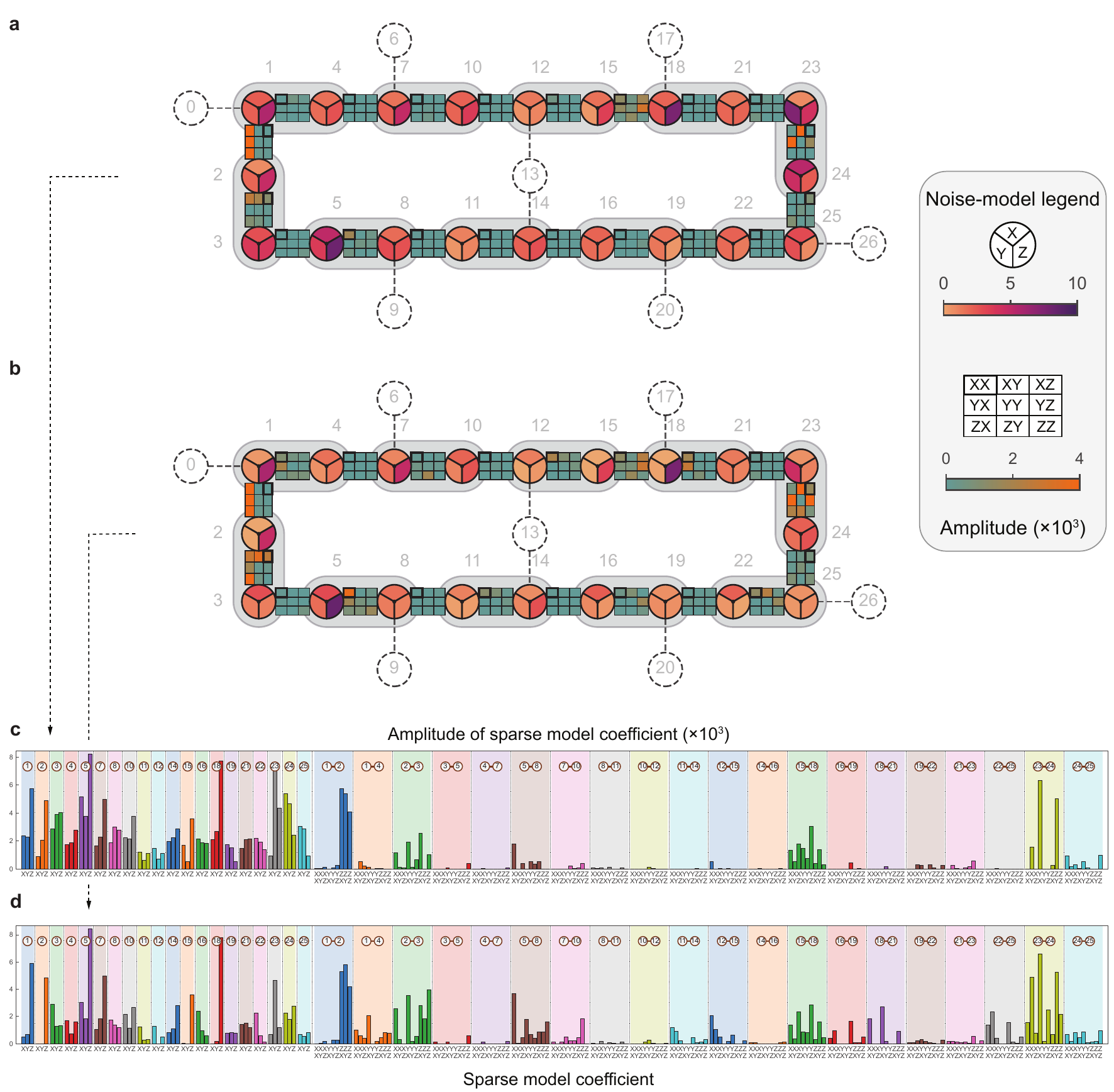}
\par\end{centering}
\caption{\textbf{Comparison between unit-depth and symmetry noise
    model fitting.}  Panels (\textbf{a}) and (\textbf{b}) depict the
  \textsc{ibm\_hanoi} processor topology along with a 20-qubit layer
  consisting of 10 {\sc{cx}} gates. We learn the noise model using
  ({\bf{a}}) unit-depth circuits and ({\bf{b}}) using a symmetry
  assumption on the noise channel. The resulting weight-one Pauli
  generators (X, Y, Z terms) in the Lindblad model are given as wedges
  inside the circular qubit nodes. Similarly, the weight-two Pauli
  generators (XX, XY, and so on) in the Lindblad model are visualized
  by the $3\times 3$ grids connecting pairs of qubits. The legend on
  the right shows the corresponding color bar detailing the noise
  amplitude. Panels (\textbf{c}) and (\textbf{d}) present the same
  data (see down-pointing arrows) as bar plot of the noise-model
  coefficients.}\label{fig:learn-large-chip-symm-unit}
\end{figure}

In Fig.~\ref{fig:dd-learning-1} we provide the full data for the
noise-model learning setup of Fig.~2a in the main text, as measured in
all nine bases determined by the learning protocol.

Finally, in Fig.~\ref{fig:learn-large-chip-symm-unit}, we compare the
noise model extracted using the unit-depth and symmetric learning
post-processing methods for a 20 qubit layer with 10 {\sc{cx}}
gates. The bottom panels of the figure show the noise-model
coefficients obtained using the unit-depth and symmetric methods,
respectively.  Aside from some localized differences, the profiles of
the two noise models were found to match well overall. For all
experiments in this work, aside from the present one, we used
symmetry-based model fitting.

\end{document}